\documentclass[a4paper,11pt]{article}
\usepackage{amsmath,amsthm,amssymb,amsfonts}
\usepackage{bm}
\allowdisplaybreaks
\usepackage{algorithm}
\usepackage{xspace}
\usepackage{xparse}
\usepackage{nicefrac}
\usepackage{hyperref}
\usepackage[capitalise]{cleveref}
\usepackage{fullpage}

\usepackage{changes}
\definechangesauthor[color=red, name={Francesco}]{F}

\newtheorem{definition}{Definition}[section]
\newtheorem{lemma}[definition]{Lemma}
\newtheorem{theorem}[definition]{Theorem}

\newtheorem{remark}{Remark}

\title{\textbf{Step-by-Step Community Detection in Volume-Regular Graphs}} 
\author{Luca Becchetti\thanks{Partially
supported by ERC Advanced Grant 788893 AMDROMA ``Algorithmic and Mechanism
Design Research in Online Markets'' and MIUR PRIN project ALGADIMAR ``Algorithms,
Games, and Digital Markets''} \\
		{\small{}Sapienza Università di Roma}\\
		{\small{}Rome, Italy}\\
		{\small{}\texttt{becchetti@diag.uniroma1.it}}\\
	\and Emilio Cruciani \\
		{\small{}Inria, I3S Lab, UCA, CNRS}\\
		{\small{}Sophia Antipolis, France}\\
		{\small{}\texttt{emilio.cruciani@inria.fr}}\\
  \and Francesco Pasquale\thanks{Partially supported by the University
of ``Tor Vergata'' under research programme ``Mission: Sustainability'' project
ISIDE (grant no.  E81I18000110005)}\\ 
		{\small{}Università di Roma Tor Vergata}\\
		{\small{}Rome, Italy}\\
		{\small{}\texttt{pasquale@mat.uniroma2.it}} 
  \and Sara Rizzo\\
		{\small{}Gran Sasso Science Institute}\\
		{\small{}L'Aquila, Italy}\\
		{\small{}\texttt{sara.rizzo@gssi.it}}\\
}
\date{}

\newcommand{\dv}{\delta} 
\newcommand{\w}{w}
\newcommand{\xerr}{\bm{e}} 
\newcommand{\Abs}[1]{\left \vert #1 \right \vert}
\newcommand{\abs}[1]{\vert #1 \vert}
\newcommand{\norm}[1]{\Vert #1 \Vert}
\newcommand{\Norm}[1]{\left \Vert #1 \right \Vert}
\newcommand{\Dsq}{D^{\frac{1}{2}}}
\newcommand{\Dsqinv}{D^{-\frac{1}{2}}}

\newcommand{\sgn}{\mathrm{sgn}}  
\newcommand{\vol}{\mathrm{vol}}

\newcommand{\Vunion}{\hat{V}}

\newcommand{\nmin}{n_{_{\min}}}
\newcommand{\nmax}{n_{_{\max}}}

\newcommand{\Pos}{i:\alpha_i\bm{v}_i(u)>0}
\newcommand{\Neg}{i:\alpha_i\bm{v}_i(u)<0}

\newcommand{\bu}{\bm{u}}
\newcommand{\bv}{\bm{v}}
\newcommand{\bw}{\bm{w}}
\newcommand{\bx}{\bm{x}}
\newcommand{\bone}{{\mathbf 1}}

\newcommand{\T}{\intercal}

\RenewDocumentCommand\Pr{gg}{%
	\ensuremath{%
		\mathbf{P} \IfNoValueTF{#1}{}{\left( #1 \IfNoValueTF{#2}{}{\cond #2} \right)}
	}
}
\NewDocumentCommand\Prob{mgg}{%
	\ensuremath{%
		\mathbf{P}_{#1} \IfNoValueTF{#2}{}{\left( #2 \IfNoValueTF{#3}{}{\cond #3} \right)}
	}
}

\NewDocumentCommand\Ex{gg}{%
	\ensuremath{%
		\mathbf{E} \IfNoValueTF{#1}{}{\left[ #1 \IfNoValueTF{#2}{}{\cond #2} \right]}
	}
}

\NewDocumentCommand\csa{gg}{%
    \IfNoValueTF{#2}{}{\ensuremath{\left(#1,\,#2\right)}-}community-sensitive algorithm}

\newcommand{\bigO}{\mathcal{O}}

\newcommand{\averaging}{\textsc{Averaging}\xspace} 
\newcommand{\avg}{\textsc{Averaging}\xspace} 
\newcommand{\avgbip}{\textsc{Averaging Bipartite}\xspace} 
\newcommand{\planted}{\mathcal{G}_{2n,p,q}}

\NewDocumentCommand\dham{gg}{\ensuremath{d_H\IfNoValueTF{#1}{}{\left(#1, #2\right)}}}

\renewcommand{\leq}{\leqslant}
\renewcommand{\le}{\leqslant}
\renewcommand{\geq}{\geqslant}
\renewcommand{\ge}{\geqslant}

\renewcommand{\epsilon}{\varepsilon}

\newcommand{\calA}{\mathcal{A}}

\begin{document}

\maketitle

\begin{abstract}
	Spectral techniques have proved amongst the most effective approaches to graph
clustering. However, in general they require explicit computation of the main
eigenvectors of a suitable matrix (usually the Laplacian matrix of the graph).

Recent work (e.g., Becchetti et al., SODA 2017) suggests that observing the
temporal evolution of the power method applied to an initial random vector may,
at least in some cases, provide enough information on the space spanned by the
first two eigenvectors, so as to allow recovery of a hidden partition without explicit
eigenvector computations. While the results of Becchetti et al.\ apply 
to perfectly balanced partitions and/or graphs that exhibit very strong 
forms of regularity, we extend their approach to graphs containing a 
hidden $k$ partition and characterized by a milder form of 
volume-regularity. We show that the class of $k$-\emph{volume regular} graphs is the largest class of
undirected (possibly weighted) graphs whose transition matrix admits $k$ ``stepwise'' eigenvectors
(i.e., vectors that are constant over each set of the hidden partition). 
To obtain this result, we highlight a connection between volume regularity and lumpability of Markov
chains. Moreover, we prove that if the stepwise eigenvectors are those
associated to the first $k$ eigenvalues and the gap between the $k$-th and the
($k$+1)-th eigenvalues is sufficiently large, the \avg{} dynamics of 
Becchetti et al.\ recovers the underlying community structure of the graph in logarithmic time, with high
probability.
\end{abstract}

\newpage

\section{Introduction}\label{sec:intro}

Clustering a graph in a way that reflects underlying community structure 
is a very important mining task~\cite{fortunato_community_2010}.
Informally speaking, in the classical setting, we are given a possibly weighted
graph $G$ and an integer $k$.  
Our goal is to partition the vertex set of $G=(V,E)$ into $k$ disjoint subsets,
so that the $k$ induced subgraphs have high inner and low outer expansion. 
Spectral techniques have proved amongst the most effective approaches to graph
clustering~\cite{ng2002spectral,shi2000normalized,von2007tutorial}. 
The general approach to spectral graph clustering~\cite{von2007tutorial} 
normally implies embedding the vertices of $G$ into the $k$-dimensional 
subspace spanned by the main $k$ eigenvectors of a matrix defined in terms 
of $G$'s adjacency matrix, typically its (normalized) Laplacian.
Intuitively, one expects that, for a well-clustered graph
with $k$ communities, the profiles of the first $k$ eigenvectors are correlated
with the underlying community structure of $G$.  Recent work has provided
theoretical support to this approach. In particular,~\cite{lee2014multiway}
showed that, given the first $k$ orthonormal eigenvectors of the normalized
Laplacian, it is possible to produce a $k$-partition of the vertex set,
corresponding to $k$ suitably-defined indicator vectors, such that the
associated values of the Rayleigh quotient are relatively small. More
recently,~\cite{PSZ17} proved that, under suitable hypotheses on the spectral
gap between the $k$-th and ($k$+1)-th eigenvalue of the normalized Laplacian of
$G$, the span of the first $k$ eigenvectors largely overlaps with the span of
$\{D^{\frac{1}{2}}\bm{g}_1,\ldots, D^{\frac{1}{2}}\bm{g}_k\}$, where 
$D$ is the diagonal degree matrix of $G$, while the 
$\bm{g}_i$'s are indicator vectors describing a $k$-way partition 
$\{S_i\}_{i=1}^k$ of $V$ such that, for every $i$, the conductance of $S_i$ 
is at most the $k$-way expansion constant $\rho(k)$~\cite{lee2014multiway}. 
Note that, if $\bm{v}$ is an eigenvector associated to the $i$-th smallest
eigenvalue of the normalized Laplacian,
$D^{-\frac{1}{2}}\bm{v}$ is an eigenvector corresponding to the $i$-th largest
eigenvalue of the random walk's transition matrix associated to $G$.
Hence, when $G$ is well-clustered, one might reasonably expect the first $k$ 
eigenvectors of $P$ to exhibit almost-``stepwise'' profiles reflecting 
$G$'s underlying community structure.
The aforementioned spectral approaches require explicit computation of the $k$
main eigenvectors of a (generally symmetric) matrix. 

In~\cite{becchetti2017find}, the authors considered the case $k = 2$ for which 
they proposed the following distributed algorithm (\averaging{} dynamics, 
\cref{alg:avg_dyn}): 
``At the outset, every node picks an initial value, independently and uniformly
at random in $\{-1,1\}$; then, in each synchronous round, every node updates 
its value to the average of those held by its neighbors. 
A node also tags itself \texttt{blue} if the last update increased its value, 
\texttt{red} otherwise''~\cite{becchetti2017find}.  
The authors showed that, under a variety of graph models exhibiting sparse
balanced cuts, including the \emph{stochastic block
model}~\cite{holland1983stochastic}, the process resulting from the above
simple local rule converges, in logarithmic time, to a coloring that, depending on the model, exactly or approximately reflects the underlying cut.
They further elaborated on how to extend the proposed approach to the case
of multiple communities, providing an analysis for a strongly regular version
of the stochastic block model with multiple communities. 
While results like those presented in~\cite{lee2014multiway,PSZ17} provide further theoretical
justification for spectral clustering, the approach proposed
in~\cite{becchetti2017find} suggests that observing the temporal evolution of
the power method applied to an initial random vector may, at least in some
cases, provide equivalent information, without requiring explicit eigenvector
computations. 

\subsection{Our contributions}
The goal of this work is to take a further step in this direction by 
considering a more general class of graphs, even if still relatively ``regular'', 
than the one considered in~\cite{becchetti2017find}.
The analysis of the \avg{} dynamics on this class is considerably harder,
but it is likely to provide insights into the challenges of analyzing the
general case, without all the intricacies of the latter. 
Our contribution is as follows:
\begin{itemize}
    \item We define the class of $k$-\emph{volume-regular} graphs. This class
    of edge-weighted graphs includes those considered in~\cite{becchetti2017find}
    and it is the largest class of undirected, possibly weighted graphs that
    admit $k$ ``stepwise'' eigenvectors (i.e., having constant values 
    over the $k$ steps that identify the hidden partition).
    This result uses a connection between volume regularity and lumpability 
    of Markov chains~\cite{kemeny1960finite,tian2006lumpability}.
    \item If the stepwise eigenvectors are those associated to the first $k$
    eigenvalues and the gap between the $k$-th and the ($k$+1)-th eigenvalues
    is sufficiently large, we show that running the \avg{} dynamics
    for a suitable number of steps allows recovery of the underlying community
    structure of the graph, with high probability.%
    \footnote{An event $\mathcal{E}_n$ holds 
    \emph{with high probability} (\emph{w.h.p.}) 
    if $\Pr{\mathcal{E}_n} = 1 - \bigO(n^{-\gamma})$, 
    for some constant $\gamma > 0$.} 
    To prove this, we provide a family of mutually orthonormal vectors which,
    when the graph is volume-regular, span the eigenspace of the main $k$ 
    eigenvectors of the normalized adjacency matrix of the graph. 
    It should be noted that the first and second of these vectors are 
    respectively the main eigenvector and the Fiedler 
    vector~\cite{fiedler1989laplacian} associated to the normalized adjacency matrix.
    \item While the results of~\cite{becchetti2017find} apply when the
    underlying communities are of the same size, our results do not require
    this assumption and they apply to weighted graphs. It should also be noted
    that volume regularity is a weaker notion than regularity of the graph. 
    \item We further show that variants of the \avg{} dynamics 
    (and/or its labeling rule) can address different 
    problems (e.g., identifying bipartiteness) and/or other graph classes.
\end{itemize}

We further note that the overall algorithm we consider can be viewed as a fully
decentralized, synchronous algorithm that works in \emph{anonymous} networks,%
\footnote{Nodes do not possess distinguished identities.} 
with a completely local clustering criterion, though it cannot be considered a
\emph{dynamics} in the sense of~\cite{becchetti2017find} since it requires a
bound on the number of nodes in the underlying network.

Finally, this paper extends a preliminary version \cite{becchetti2019step} in several ways. 
To begin, the main result presented in \cite{becchetti2019step} was weaker, in the sense
that the constraints imposed on the eigenvalues in \cite[Theorem~9]{becchetti2019step} polynomially
depend on network parameters like the maximum degree and the number of vertices.
In this respect, they are substantially stronger than those imposed to prove \cref{thm:main}, where
results (in particular, the time window in which recovery of the hidden partition is
possible) are expressed in terms of the spectrum of the graph, while constraints imposed
on the second eigenvalue only logarithmically depend on the aforementioned network
parameters. In reframing these results, we also realized that the presence of a window in which
recovery is possible is something that is hardly avoidable in general using the simple averaging
heuristic of \cite{becchetti2017find}. This is something we remark 
right after \cref{thm:main} (see \cref{rmk1}), while  we also 
observe (see \cref{rmk:equa_sized,rmk:two_com}) that the analysis presented here also encompasses the class of
regular graphs considered in \cite{becchetti2017find} as a special case, something that was not obvious
in \cite{becchetti2019step}.
Finally, the result given in \cite{becchetti2019step} for bipartite graphs assumed 
volume regularity, an assumption that is not necessary as we show in \cref{sec:bipartite}.

\subsection{Further related work}\label{ssec:related}
We briefly discuss further work that bears some relationship to this paper,
either because it adopts simple and/or decentralized heuristics to uncover
community structure, or because it relies on the use of spectral techniques.

\paragraph{Decentralized heuristics for block reconstruction.}
\emph{Label propagation algorithms} \cite{raghavan2007near} 
are dynamics based on majority updating rules~\cite{AAE07} and have been applied
for detecting communities in complex networks.
Several papers present experimental results for such protocols on specific
classes of clustered graphs~\cite{barber2009detecting,liu2010advanced,
raghavan2007near}. 
The only available rigorous analysis of a label propagation algorithm on planted
partition graphs is the one presented in~\cite{kothapalli2013analysis}, where
the authors analyze a label propagation algorithm on $\planted$ graphs 
in the case of dense topologies.
In particular, their analysis considers the case where
$p = \Omega(1/n^{\frac{1}{4}-\epsilon})$ and $q = \bigO(p^2)$, 
a parameter range in which very dense clusters of constant diameter separated 
by a sparse cut occur w.h.p.
In this setting, characterized by a polynomial gap between $p$ and
$q$, simple combinatorial and concentration arguments show that the protocol
converges in constant expected time.
A logarithmic bound for sparser topologies is conjectured in~\cite{kothapalli2013analysis}.

Following~\cite{becchetti2017find}, a number of recent papers analyze simple 
distributed  algorithms  for community detection that rely on elementary dynamics.
In the \avg dynamics considered in this paper, every node communicates in 
parallel with all its neighbors in each round. 
While this might be too expensive in scenarios characterized by dense topologies,
it is simply infeasible in other settings (for instance, when links represent 
opportunistic meetings that occur asynchronously). 
Motivated by similar considerations, a first line of follow-up work 
considered ``sparsified'', asynchronous variants of the \avg 
dynamics~\cite{BCMNPRT18,TMM18,SZ18}. 
  
Another interesting direction is the rigorous analysis of well-known (non-linear)
dynamics based on \emph{majority rules} on graphs that exhibit community 
structure. 
In~\cite{CNNS18}, Cruciani et al.\ consider the \emph{2-Choices} 
dynamics where, in each round, every node picks two random neighbors 
and updates its value to the most frequent among its value and those held by 
its sampled neighbors. 
They show that if the underlying graph has a suitable core-periphery 
structure and the process starts in a configuration where nodes in core and 
periphery have different states, the system either rapidly converges to the
core's state or reaches a metastable regime that reflects the underlying 
graph structure.
Similar results have been also obtained for clustered regular graphs 
with dense communities in~\cite{CNS18}, where the \emph{2-Choices} dynamics 
is proposed as a distributed algorithm for community detection.

Although based on the \avg dynamics and thus extremely simple and fully 
decentralized, the algorithm we consider in this paper is not itself a dynamics 
in the sense proposed in~\cite{becchetti2017find}, since its clustering 
criterion is applied within a time window, which in turn requires (at 
least approximate) knowledge of the network size.

Because of their relevance for the reconstruction problem, we also 
briefly discuss the class of \emph{belief propagation algorithms},
best known as message-passing algorithms for performing inference in graphical 
models~\cite{MAC03}. 
Though not a dynamics, belief propagation is still a simple approach. 
Moreover, there is non-rigorous, strong supporting evidence that some 
\emph{belief propagation algorithms} might be optimal for the reconstruction
problem~\cite{decelle_asymptotic_2011}.
A rigorous analysis is a major challenge; in particular, convergence to
the correct value of belief propagation is far from being fully-understood on
graphs which are not trees~\cite{MK07,WEI00}.
As we discuss in the next subsection, more complex algorithms inspired by 
belief propagation have been rigorously shown to perform reconstruction 
optimally.

\paragraph{General algorithms for block reconstruction.}
Several algorithms for community detection are \textit{spectral}: They typically 
consider the eigenvector associated to the second largest eigenvalue of the
adjacency matrix $A$ of $G$, or the eigenvector corresponding to the largest 
eigenvalue of the matrix $A- \frac dn J$~\cite{boppana1987eigenvalues,
coja-oghlan_spectral_2005,coja2010graph,mcsherry2001spectral},%
\footnote{$A$ is the adjacency matrix of $G$, $J$ is the matrix having all 
entries equal to $1$, $d$ is the average degree, and $n$ is the number of 
vertices.}
since these are correlated with the hidden partition.  
More recently spectral algorithms have been proposed~\cite{abbe2015detection,
BLM15,coja2010graph,krzakala2013spectral,mossel_proof_2013,PSZ17} 
that find a weak reconstruction even in the sparse, tight regime. 

Interestingly, spectral algorithms turn out to be a feasible approach 
also in distributed settings.
In particular, Kempe and McSherry~\cite{kempe2004decentralized}
show that eigenvalue computations can be performed in a distributed fashion,
yielding distributed algorithms for community detection under various models,
including the stochastic block model. 
However, their algorithm does not match any simple decentralized computing model.
In particular, the algorithm of Kempe and McSherry as well as any distributed 
version of the above mentioned centralized algorithms are neither dynamics, 
nor do they correspond to the notion of \emph{light-weight} algorithm of 
Hassin and Peleg~\cite{hassin2001distributed}. 
Moreover, the mixing time of the simple random walk on the graph is a
bottleneck for the distributed algorithm of Kempe and McSherry and for any
algorithm that performs community detection in a graph $G$ by employing the
power method or the Lanczos method~\cite{lanczos1950iteration} as a subroutine.
This is not the case for the \avg dynamics, since it removes the component of 
the state in the span of the main eigenvector.

In general, the reconstruction problem has been studied extensively using a 
multiplicity of techniques, which include 
combinatorial algorithms~\cite{dyer1989solution}, 
belief propagation~\cite{decelle_asymptotic_2011} 
and variants of it~\cite{MNS14}, 
spectral-based techniques~\cite{coja2010graph,mcsherry2001spectral},
Metropolis approaches~\cite{jerrum_metropolis_1998},
and semidefinite programming~\cite{abbe2014exact}, among others.

\subsection{Roadmap}
The rest of this paper is organized as follows. In \cref{sec:preli}, we
formally define the \avg dynamics and briefly recall how it is connected with
the transition matrix of a random walk on the underlying graph. We also define
the notion of \textit{community-sensitive algorithm} and the class of
\textit{clustered volume-regular graphs}. 
In \cref{sec:lump} we show the relation between lumpability of Markov chains and
volume-regular graphs. In \cref{sec:mainresult} we state the
main result of the paper (see \cref{thm:main}) on the analysis of the \avg for
clustered volume-regular graphs: We give the two main technical lemmas and show
how the main theorem derives from them. In  
\cref{sec:bipartite}, we show how slightly modified versions of the \avg{} 
dynamics can be used to identify the hidden partition of other non-clustered 
volume-regular graphs, e.g., bipartite graphs.
In \cref{sec:concl} we briefly show how our approach can be extended to 
slightly more general graph classes than the ones considered in this 
paper. We finally highlight some open problems and directions for further 
research on the topic. 

\section{Preliminaries}\label{sec:preli}
\paragraph{Notation.}
Consider an undirected edge-weighted graph $G=(V,E,w)$ with nonnegative weights.
For each node $u \in V$, we denote by $\dv(u)$ the \emph{volume}, 
or \emph{weighted degree}, of node $u$, namely
\(
    \dv(u) = \sum_{v:(u,v)\in E} \w(u,v).
\)
Similarly, we denote the volume of a set of nodes $T \subseteq V$ as 
$\vol(T) := \sum_{u \in T} \dv(u)$.
$D$ denotes the diagonal matrix, such that $D_{uu} = \dv(u)$ for each $u \in V$.
Without loss of generality we assume $\min_u \dv(u) = 1$, 
since the behavior of the \averaging{} dynamics (and the corresponding 
analysis) is not affected by a normalization of the weights.
We refer to the maximum volume of a node as $\Delta := \max_u \dv(u)$.

In the remainder, $W$ denotes the \emph{weighted adjacency matrix} of 
$G$, while $P = D^{-1}W$ is the \emph{transition matrix} of a random 
walk on $G$, in which a transition from node $u$ to node $v$ occurs 
with probability proportional to $\w(u,v)$.
We call $\lambda_1,\ldots,\lambda_n$ the eigenvalues of $P$, 
in non-increasing order,
and $\bm{v}_1,\ldots,\bm{v}_n$ a family of eigenvectors of $P$, such 
that $P\bm{v}_i = \lambda_i \bm{v}_i$.
We let $N=\Dsqinv W\Dsqinv =\Dsq P\Dsqinv$ denote the \emph{normalized weighted
adjacency matrix} of $G$.
Note that $N$ is real and symmetric (thus, the eigenvectors of $N$ are orthogonal)
and that its spectrum is the same as that of $P$.
We denote by $\bm{w}_1,\ldots,\bm{w}_n$ a family of eigenvectors of 
$N$, such that $N \bm{w}_i = \lambda_i \bm{w}_i$.
It is important to note that $\bm{w}_i$ is an eigenvector of $N$ if and only if $\Dsqinv{}\bm{w}_i$
is an eigenvector of $P$.

We use the Bachmann--Landau asymptotic notation 
(i.e., $\omega, \Omega, \Theta, \bigO, o$) 
to describe the limiting behavior of functions depending on $n$. 
In this sense, our results only hold for large $n$.
We say that an event $\mathcal{E}_n$ holds \emph{with high probability} 
(\emph{w.h.p.}, in short) if $\Pr{\mathcal{E}_n} = 1 - \bigO(n^{-\gamma})$,
for any positive constant $\gamma$.

\subsection{Averaging dynamics}
The simple algorithm we consider in this paper, named \averaging{} dynamics 
(\cref{alg:avg_dyn}) after \cite{becchetti2017find} in which the algorithm was 
first proposed, can be seen as an application of the power method, 
augmented with a Rademacher initialization and a suitable labeling 
scheme. In this form, it is best described as a distributed process, 
executed by the nodes of an underlying edge-weighted graph.
The \averaging{} dynamics can be used as
a building-block to achieve ``community detection'' in some classes
of ``regular'' and ``almost regular'' graphs.  Herein, we extend its 
use and analysis to broader graph classes and, in one case, to a different 
problem.
\begin{center}
\begin{algorithm}[!ht]
\caption{\averaging{} dynamics}
\begin{minipage}{0.95\textwidth}
  \begin{description}
    \item[Rademacher initialization:] At round $t = 0$, every node $v \in V$ 
    independently samples its value $\bx^{(0)}(v)$ from $\{- 1, +1\}$ uniformly at random.
    \item[Update rule:] At each subsequent round $t \geq 1$, every node $v \in V$:
    \begin{enumerate} 
      \item \emph{Averaging}: updates its value $\bx^{(t)}(v)$ to the 
      weighted average of the values of its neighbors at the end of the previous round.
      \item \emph{Labeling}: if $\bx^{(t)}(v) \geqslant \bx^{(t-1)}(v)$ 
      then $v$ sets $\texttt{label}^{(t)}(v) = 1$;
      otherwise $v$ sets $\texttt{label}^{(t)}(v) = 0$.
    \end{enumerate}
  \end{description}
\end{minipage}
\label{alg:avg_dyn}
\end{algorithm}
\end{center}

\paragraph{Spectral decomposition of the transition matrix.}
Let $\bx^{(t)}$ denote the {\em state vector} at time $t$, i.e., the 
vector whose $u$-th entry is the value held by node $u$ at time $t$. We 
let $\bx^{(0)} = \bx$ denote the initial state vector.
Globally, the \emph{averaging} update rule of \cref{alg:avg_dyn} 
corresponds to one iteration of the power method, in this case an 
application of the transition matrix $P$ to the current state vector, i.e.,
$\bx^{(t)} = P\bx^{(t-1)}$. We can write
\[
    \bm{x}^{(t)} 
    = P^t \bm{x}
    = \Dsqinv N^t \Dsq \bm{x}
    \stackrel{(a)}{=} \Dsqinv \sum_{i=1}^n \lambda_i^t \bm{w}_i \bm{w}_i^\T \sum_{i=1}^n \beta_i \bm{w}_i
    \stackrel{(b)}{=} \sum_{i=1}^n \lambda_i^t \beta_i \Dsqinv \bm{w}_i,
\]
where in $(a)$ we spectrally decomposed the matrix $N^t$ and expressed 
the vector $\Dsq \bm{x}$ as a linear combination of the eigenvectors of $N$, 
i.e., $\Dsq \bm{x} = \sum_{i=1}^n \beta_i \bm{w}_i$,
with $\beta_i = \langle \Dsq \bm{x}, \bm{w}_i \rangle$;
in $(b)$ we used that the eigenvectors of $N$ are orthonormal, 
i.e., that $\bm{w}_i^\T \bm{w}_i = 1$ for every $i \in \{1,\ldots,n\}$ 
and that $\bm{w}_i^\T \bm{w}_j = 0$ for every $i,j \in \{1,\ldots,n\}$ 
and such that $i\neq j$.
By explicitly writing the $\beta_i$s and by noting that
$\bm{w}_i = \frac{\Dsq \bm{v}_i}{\norm{\Dsq \bm{v}_i}}$
we conclude that
\begin{equation}\label{eq:state_decomp}
  \bm{x}^{(t)} 
  = \sum_{i=1}^n \lambda_i^t \frac{\langle \Dsq \bm{x}, \Dsq \bm{v}_i \rangle}{\norm{\Dsq \bm{v}_i}} \Dsqinv \frac{\Dsq \bm{v}_i}{\norm{\Dsq \bm{v}_i}}
  = \sum_{i=1}^n \lambda_i^t \alpha_i \bm{v}_i,
\end{equation}
where $\alpha_i := 
\frac{\langle \Dsq \bm{x}, \Dsq \bm{v}_i\rangle}{\norm{\Dsq \bm{v}_i}^2} 
= \frac{\bm{x}^\T D \bm{v}_i}{\norm{\Dsq \bm{v}_i}^2}$
is the length of the projection of $\Dsq \bm{x}$ on $\Dsq \bm{v}_i$.

Note that $\lambda_1 = 1$ and $\bm{v}_1 = \bm{1}$,\footnote{Here and in 
the remainder, $\bm{1}$ denotes the vector whose entries are $1$.}
since $P$ is stochastic, and $\lambda_i \in(-1,1)$ for every $i>1$,
if $G$ is connected and non bipartite.
The long term behavior of the dynamics can be written as
\[
  \lim_{t \rightarrow \infty} \bm{x}^{(t)} 
  = \lim_{t \rightarrow \infty} \sum_{i=1}^n \lambda_i^t \alpha_i \bm{v}_i 
  = \alpha_1 \bm{1},
  \quad
  \text{with }
  \alpha_1 = \frac{\sum_{u \in V} \dv(u) \bm{x}(u)}{\sum_{u \in V} \dv(u)}
  = \sum_{u \in V} \frac{\dv(u)}{\vol(V)} \bm{x}(u),
\]
i.e., each node converges to the initial global weighted average of the network.

\subsection{Community-sensitive algorithms}
We give the following definition of \textit{community-sensitive algorithm},
that closely resembles that of locality-sensitive
hashing (see, e.g.,~\cite{rajaraman2011mining}).
\begin{definition}[Community-sensitive algorithm]\label{def:comm_sen_alg}
Let $\mathcal{A}$ be a randomized algorithm that takes in input a (possibly
weighted) graph $G = (V, E)$ with a hidden partition $\mathcal{V} = \{V_1,
\dots, V_k\}$ and assigns a Boolean value $\mathcal{A}(G)[v] \in \{0,1\}$ to
each node $v \in V$. We say $\mathcal{A}$ is an \csa{\varepsilon}{\delta}, for
some $\varepsilon, \delta > 0$, if the following two conditions hold:
\begin{enumerate}
\item For each set $V_i$ of the partition and for each pair of nodes $u,v \in
V_i$ in that set, the probability that the algorithm assigns the same Boolean 
value to $u$ and $v$ is at least $1-\varepsilon$:
\[
\forall i \in [k], \, \forall u,v \in V_i, \; 
\Prob{}{\mathcal{A}(G)[u] = \mathcal{A}(G)[v]} \geqslant 1 - \varepsilon.
\]
\item For each pair $V_i, V_j$ of distinct sets of the partition and for each
pair of nodes $u \in V_i$ and $v \in V_j$, the probability that the algorithm
assigns the same value to $u$ and $v$ is at most $\delta$:
\[
\forall i,j \in [k] \, \mbox{ with } i \neq j, \, \forall u \in V_i, \forall v
\in V_j,\; \Prob{}{\mathcal{A}(G)[u] = \mathcal{A}(G)[v]} \leqslant \delta.
\]
\end{enumerate}
\end{definition}

\noindent For example, for $(\varepsilon, \delta) = (1/n, 1/2)$, an algorithm
that simply assigns the same value to all nodes would satisfy the first
condition but not the second one, while an algorithm assigning $0$ or $1$ to
each node with probability $1/2$, independently of the other nodes, would
satisfy the second condition but not the first one.

Note that \cref{alg:avg_dyn} is a distributed algorithm that, at each round
$t$, assigns one out of two labels to each node of a graph. In the next section
(see \cref{thm:main}) we  prove that a time window $[T_1, T_2]$ exists, such
that for all rounds $t \in [T_1, T_2]$, the assignment of the \averaging{}
dynamics satisfies both conditions in \cref{def:comm_sen_alg}: The first
condition with $\varepsilon = \varepsilon(n) = \bigO(n^{-\frac{1}{2}})$, the
second with $\delta = \delta(n) = 1 - \Omega(1)$.

\paragraph{Community-sensitive labeling.} We here generalize the concept of 
\emph{community-sensitive labeling} (appeared in~\cite[Definition~3]{BCMNPRT18}),
given only for the case of two communities, to the case of multiple communities. 
If we execute $\ell = \Theta(\log n)$ independent runs of an
$(\varepsilon, \delta)$-community-sensitive algorithm $\mathcal{A}$, each node
is assigned a \textit{signature} of $\ell$ binary values, with pairwise Hamming
distances probabilistically reflecting community membership of the nodes. More
precisely, let $\calA$ be an \csa{\varepsilon}{\delta} and let $\calA_1, \dots,
\calA_\ell$ be $\ell = \Theta(\log n)$ independent runs of $\calA$. For each
node $u \in V$, let $\bm{s}(u) = (s_1(u), \dots, s_\ell(u))$ denote the
\textit{signature} of node $u$, where $s_i(u) = \calA_i(G)[u]$.  For each pair
nodes $u,v$, let $h(u,v) = \abs{\{i \in [\ell] \,:\, s_i(u) \neq s_i(v)\}}$ be
the Hamming distance between $\bm{s}(u)$ and $\bm{s}(v)$.

\begin{lemma}[Community-sensitive labeling]
	Let $\calA$ be an \csa{\varepsilon}{\delta} with 
	$\varepsilon = \mathcal{O}(\frac{1}{n^{\gamma}})$ 
	for any arbitrarily small positive constant $\gamma$, 
	and $\delta = 1 - \Omega(1)$. 
	Let $\ell = \Theta(\log n)$, 
	$\alpha = \Omega(\frac{1}{n^{\gamma - c}})$ with $c \in (0,\gamma]$, 
	and $\beta = b(1-\delta)$ for any constant $b \in (0,1)$ 
	and such that $0 \leq \alpha \leq \beta \leq 1$.
	Then, for each pair of nodes $u, v \in V$ it holds that: 
	\begin{enumerate}
		\item If $u$ and $v$ belong to the same community then 
		$h(u,v) < \alpha \ell$, w.h.p.
		\item If $u$ and $v$ belong to different communities then 
		$h(u,v) \geq \beta \ell$, w.h.p.
	\end{enumerate}
\end{lemma}
\begin{proof}
	From the definition of \csa{\varepsilon}{\delta} we have that, if $u$ and $v$ belong to the same community, then 
	$\Ex{h(u,v)} = \sum_{i=1}^\ell \Prob{}{s_i(u) \neq s_i(v)}\leqslant
	\varepsilon \ell$ . Similarly, if they belong to different communities, then 
	$\Ex{h(u,v)} = \sum_{i=1}^\ell \Prob{}{s_i(u) \neq s_i(v)} \geqslant (1 - \delta) \ell$. 
	If $u$ and $v$ belong to the same community, we compute $\Prob{}{h(u,v) > \alpha \ell}$ 
	and by Markov inequality we get that
	\begin{equation}\label{eq:LB_hamming}
	\Prob{}{h(u,v) \geq \alpha \ell} \leqslant
	\frac{\Ex{h(u,v)}}{\alpha \ell} \leqslant \frac{\varepsilon}{\alpha} 
	= \mathcal{O}\left(\frac{1}{n^{c}}\right),
	\end{equation}
	where in the last inequality we use the hypothesis $\varepsilon = \mathcal{O}(\frac{1}{n^{\gamma}})$
	and $\alpha = \Omega(\frac{1}{n^{\gamma-c}})$. 
	On the other hand, if $u$ and $v$ belong to different communities, 
	we apply \cref{thm_chernoff_ext} to $h(u,v)$ by using the lower bound on the expected value of $h(u,v)$ and the hypothesis $\ell = \Theta(\log n)$. Thus,
	\begin{equation}\label{eq:UB_hamming}
	\Prob{}{h(u,v) < b(1 - \delta) \ell} \leqslant
	\exp\left( - \frac{(1-b)^2}{2}(1 - \delta)\ell\right) = \mathcal{O}\left(\frac{1}{n^d}\right),
	\end{equation}
	where $d$ is a positive constant. 
	The thesis follows by combing \cref{eq:LB_hamming,eq:UB_hamming}.
\end{proof}

\subsection{Volume-regular graphs}
Recall that, for an undirected edge-weighted graph $G = (V, E, w)$, 
we denote by $\dv(u)$ the volume a node $u \in V$,
i.e., $\dv(u) = \sum_{v:(u,v)\in E} \w(u,v)$.
Note that the transition matrix $P$ of a random walk on $G$ is such that 
$P_{uv} = w \left(u,v\right) / \delta(u)$.
Given a partition ${\mathcal V} = \{V_1,\ldots , V_k\}$ of the set of nodes
$V$, for a node $u \in V$ and a partition index $i \in [k]$, 
$\dv_i(u)$ denotes the overall weight of edges connecting $u$ to nodes in $V_i$,
\(
  \dv_i(u) = \sum_{v \in V_i \,:\, {u, v} \in E}\w\left({u, v}\right).
\)
Hence, $\dv(u) = \sum_{i=1}^k \dv_i(u)$.

\begin{definition}[Volume-regular graph]\label{def:vol_reg}
Let $G = (V, E, w)$ be an undirected edge-weighted graph with $|V| = n$ nodes
and let ${\mathcal V} = \{V_1,\ldots, V_k\}$ be a $k$-partition of the nodes,
for some $k \in [n]$. We say that $G$ is \emph{volume regular} with respect to
$\mathcal{V}$ if, for every pair of partition indexes $i, j \in [k]$ and for
every pair of nodes $u, v \in V_i$, 
\(
\frac{\dv_j(u)}{\dv(u)} = \frac{\dv_j(v)}{\dv(v)}.
\)
We say that $G$ is $k$-volume regular if there exists a $k$-partition 
$\mathcal{V}$ of the nodes such that $G$ is volume regular with respect to 
$\mathcal{V}$.
\end{definition}
In other words, $G$ is volume regular if there exists a partition of the nodes 
such that the fraction of a node's volume toward a set of the partition is 
constant across nodes of the same set. Note that all graphs with $n$ nodes are
trivially $1$- and $n$-volume regular.

Let $G = (V, E, w)$ be a $k$-volume regular graph and let $P$ be the transition 
matrix of a random walk on $G$. In the next lemma we prove that
the span of $k$ linearly independent eigenvectors of $P$ equals
the span of the indicator vectors of the $k$ communities of $G$.
The proof makes use of the correspondence between random walks on volume regular
graphs and \textit{ordinary lumpable} Markov chains~\cite{kemeny1960finite};
in particular the result follows from \cref{le:lum_reg} and \cref{le:lum_step}, 
that we prove in \cref{sec:lump}.
\begin{lemma}\label{lem:span_stepwise}
Let $P$ be the transition matrix of a random walk on a $k$-volume regular graph
$G = (V,E,w)$ with $k$-partition $\mathcal{V} = \{V_1, \dots, V_k\}$.  
There exists a family $\{\bm{v}_1, \ldots, \bm{v}_k\}$ of linearly independent
eigenvectors of $P$ such that 
\(
  Span\left( \{\bm{v}_1, \dots, \bm{v_k} \}\right) 
  = Span\left( \{\bone_{V_1}, \dots, \bone_{V_k} \} \right),
\)
with $\bone_{V_i}$ the indicator vector of the
$i$-th set of the partition, for $i \in [k]$.
\end{lemma}

In the rest of the paper we call ``stepwise'' the eigenvectors of $P$ that
can be written as linear combinations of the indicator vectors of the
communities. In the next definition, we formalize the fact that a
$k$-volume regular graph is \textit{clustered} if the $k$ linearly independent
stepwise eigenvectors of $P$, whose existence is guaranteed by the above lemma,
are associated to the $k$ largest eigenvalues of $P$.

\begin{definition}[Clustered volume regular graph]\label{def:clust_vol_reg}
Let $G = (V,E,w)$ be a $k$-volume regular graph and let $P$ be the transition
matrix of a random walk on $G$.  We say that $G$ is a \emph{clustered
$k$-volume regular graph} if the $k$ stepwise eigenvectors of $P$ are
associated to the first $k$ largest eigenvalues of $P$.
\end{definition}

\section{Volume-regular graphs and lumpable Markov chains}\label{sec:lump}

The class of volume-regular graphs is deeply connected with the definition of
\textit{lumpability}~\cite{kemeny1960finite} of Markov chains. We here first
recall the definition of lumpable Markov chain and then show that a graph $G$
is volume-regular if and only if the associated weighted random walk is a
lumpable Markov chain.

\begin{definition}[Ordinary lumpability of Markov Chains]\label{def:lumpable}
Let $\{X_t\}_t$ be a finite Markov chain with state space $V$ and transition
matrix $P = (P_{uv})_{u,v \in V}$ and let $\mathcal{V} = \{V_1, \ldots, V_k\}$
be a partition of the state space.  Markov chain $\{X_t\}_t$ is \emph{ordinary
lumpable} with respect to $\mathcal{V}$ if, for every pair of partition indexes
$i, j \in [k]$ and for every pair of nodes in the same set of the partition $u,
v \in V_i$, it holds that
\begin{equation}\label{eq:lump_cond}
	\sum_{w \in V_i} P_{uw} = \sum_{w \in V_i} P_{vw},
	\quad
	\forall~u,v \in V_j.
\end{equation}	
We define the \emph{lumped matrix} $\widehat{P}$ of the Markov Chain 
as the matrix such that $\widehat{P}_{ij} = \sum_{w \in V_i} P_{uw}$,
for any $u \in V_j$.
\end{definition}

We first prove that random walks on Volume-regular graphs define 
exactly the subset of reversible and ordinary lumpable Markov chains.
\begin{lemma}\label{le:lum_reg}
A reversible Markov chain $\{X_t\}_t$ is ordinary lumpable if and only if 
it is a random walk on a volume-regular graph.
\label{thm:comreg-lumpable}
\end{lemma}
\begin{proof}
Assume first that $\{X_t\}_t$ is ordinary lumpable and let $P$ be the corresponding transition matrix. 
Consider the weighted graph $G = (V, E, w)$ obtained from $P$ as follows: 
$V$ corresponds to the set of states in $P$, while $\w(u,v) = \pi(u)P_{uv}$,
for every $u, v\in V$, with $\pi$ the stationary distribution of $P$.
Note that $G$ is an undirected graph, i.e., 
\(
	\w(u,v) = \pi(u)P_{uv} 
	\stackrel{(a)}{=} \pi(v)P_{vu} = \w(v,u),
\)
where $(a)$ holds because $P$ is reversible.
Moreover
\[
	\dv(u) = \sum_{z \in V} \w(u,v) = \sum_{z \in V} \pi(u)P_{uv} 
	= \pi(u) \sum_{z \in V} P_{uv} \stackrel{(a)}{=} \pi(u),
\]
where $(a)$ holds because $P$ is stochastic.
Thus $G$ meets \cref{def:vol_reg} because, for any $u,v \in V_i$,
\[
	\frac{\dv_j(u)}{\dv(u)} = \frac{1}{\pi(u)} \sum_{z \in V_j} \w(u,z)
	= \sum_{z \in V_j} P_{uz}
	= \sum_{z \in V_j} P_{vz}
	= \frac{1}{\pi(v)} \sum_{z \in V_j} \w(v,z) = \frac{\dv_j(v)}{\dv(v)}.
\]

Next, assume $G$ is $k$-volume-regular with respect to the partition 
${\mathcal V} = \{V_1,\ldots , V_k\}$. Let $P$ be the transition 
matrix of the corresponding random walk. For every $i, j\in [k]$ 
and for every $u, v\in V_i$ we have:
\begin{align*}
	&\sum_{z\in V_j}P_{uz} 
	= \sum_{z\in V_j} \frac{\w(u, z)}{\dv(u)} 
	= \frac{\dv_j(u)}{\dv(u)} 
	\stackrel{(a)}{=} \frac{\dv_j(v)}{\dv(v)} 
	= \sum_{z\in V_j} \frac{\w(v, z)}{\dv(v)} 
	= \sum_{z\in V_j}P_{vz},
\end{align*}
where $(a)$ follows from \cref{def:vol_reg}. 
Moreover note that $P$ is reversible with respect to distribution
$\pi$, where $\pi(u) = \frac{\dv(u)}{\vol(G)}$.
\end{proof}

Note that infinitely many $k$-volume-regular graphs have the same 
$k$-ordinary lumpable random walk chain.

We next show that a Markov chain is $k$-ordinary lumpable if and 
only if the corresponding transition matrix $P$ has $k$ stepwise, linearly independent eigenvectors.
\begin{lemma}\label{le:lum_step}
Let $P$ be the transition matrix of a Markov chain. 
Then $P$ has $k$ stepwise linearly independent eigenvectors
if and only if $P$ is ordinary lumpable.
\end{lemma}
\begin{proof}
We divide the proof in two parts. 
First, we assume that $P$ is ordinary lumpable
and show that $P$ has $k$ stepwise linearly independent eigenvectors.
Second, we assume that $P$ has $k$ stepwise linearly independent eigenvectors 
and show that $P$ is ordinary lumpable.

\medskip\noindent 1.
Let $P$ be ordinary lumpable and $\widehat{P}$ its lumped matrix.
Let $\lambda_i, \bv_i$ be the eigenvalues and eigenvectors of $\widehat{P}$,
for each $i \in [k]$.
Let $\bw_i \in \mathbb{R}^n$ be a stepwise vector defined as
\[
\bw_i =
(\underbrace{\bv_i(1),\ldots,\bv_i(1)}_\text{$\abs{V_1}$ times},\,
\underbrace{\bv_i(2),\ldots,\bv_i(2)}_\text{$\abs{V_2}$ times},\,
\ldots,\,
\underbrace{\bv_i(k),\ldots,\bv_i(k)}_\text{$\abs{V_k}$ times})^\intercal,
\]
where $\bv_i(j)$ indicates the $j$-th component of $\bv_i$,
and then the $\abs{V_j}$ components relative to $V_j$ are all equal to $\bv_i(j)$.

Since the eigenvectors $\bv_i$ of $\widehat{P}$ are linearly independent,
the vectors $\bw_i$ are also linearly independent.
Moreover, it is easy to see that $P\bw_i = \lambda_i \bw_i$
by just verifying the equation for every $i \in [k]$.

\medskip\noindent 2.
Assume $P$ has $k$ stepwise linearly independent eigenvectors $\bw_i$, 
associated to $k$ eigenvalues $\lambda_i$, for each $i \in [k]$.
Let $\bv_i \in \mathbb{R}^k$ the vector that has as components
the $k$ constant values in the steps of $\bw_i$. 
Since the $\bw_i$ are linearly independent, the $\bv_i$ also are.

For every eigenvector $\bw_i$
and for every two states $x,y \in V_l$, for every $l\in [k]$,
we have that $\lambda_i\bw_i(x) = \lambda_i\bw_i(y)$
since $\bw_i$ is stepwise. 
Then, since $P\bw_i = \lambda_i\bw_i$, we have that
\[
\sum_{j=1}^{k}{\sum_{z \in V_j} P_{xz}} \bv_i(j) = (P\bw_i)(x)
= (P\bw_i)(y) = \sum_{j=1}^{k}{\sum_{z \in V_j} P_{yz}} \bv_i(j).
\]
Thus
\(
\sum_{j=1}^{k} \bv_i(j) {\sum_{z \in V_j} \left(P_{xz} - P_{yz}\right)} = 0
\)
and then it follows that
\[
\sum_{j=1}^{k} \bv_i(j) \, \bu_{xy}(j) =
\langle \bu_{xy}, \bv_i \rangle = 0,
\]
where $\bu_{xy}(j) := \sum_{z \in V_j} \left(P_{xz} - P_{yz}\right)$.
Since the $\bv_i$ are $k$ linearly independent vectors in a $k$-dimensional space,
$\bu_{xy}$ cannot be orthogonal to all of them 
and then it has to be the null vector, 
i.e., $\bu_{xy}(j) = 0$ for all $j \in [k]$.
This implies that $P$ is ordinary lumpable, 
i.e., $\sum_{z \in V_j} P_{xz} = \sum_{z \in V_j} P_{yz}$.
It is easy to verify that the eigenvalues and eigenvectors of $\widehat{P}$
are exactly $\lambda_i, \bv_i$, with $i \in [k]$.
\end{proof}

\section{Averaging dynamics on clustered volume regular graphs}\label{sec:mainresult}
Let $\nmin := \min_{i \in [k]} |V_i|$ and $\nmax := \max_{i \in [k]} |V_i|$ be
the maximum and minimum sizes of the communities of a volume-regular graph $G =
(V, E, w)$ with $n$ nodes and $k$-partition $\mathcal{V} = \{V_1, \ldots,
V_k\}$. Recall also that $\Delta$ is the maximum weighted degree of the nodes
of $G$ and $\lambda_1, \dots, \lambda_n$ are the eigenvalues of the
transition matrix of a random walk on $G$ (see Section~\ref{sec:preli}). In
this section we prove the following result. 

\begin{theorem}\label{thm:main}
Let $G = (V, E, w)$ be a connected clustered $k$-volume-regular graph 
with $n$ nodes and $k$-partition $\mathcal{V} = \{V_1,\ldots,V_k\}$, such that
$\Delta\leq\frac{\sqrt{\nmin}}{25}$
and $2 \Delta (\nmax/\nmin) < k \leq \sqrt{n}$. Assume further that 
$1-\lambda_2\leq \frac{\lambda_k \log(\lambda_k/\lambda_{k+1})}{7\log(2\Delta n)}$
and $\lambda_k \geq \frac{7\lambda_2-5}{2}$.
A non-empty time interval $[T_1, T_2]$ exists,
with $T_1 = \bigO\left(\frac{\log n}{\log(\lambda_k / \lambda_{k+1})}\right)$ 
and $T_2 = \Omega\left( \frac{\lambda_k}{1-\lambda_2} \right)$, 
such that for each $t \in [T_1, T_2]$,  the \averaging{} dynamics truncated at 
round $t$ is a $(\bigO(n^{-1/2}),\,1 - \Omega(1))$-community-sensitive algorithm.
\end{theorem}

\begin{remark}[The extent of the time-window]\label{rmk1}
Notice that the time window cannot be too long: by Cheeger's inequality
$1-\lambda_2\ge\frac{h_G^2}{2} \geq 1 / (2\Delta^2 n^2)$,%
\footnote{This can be seen by observing that: $i)$ the minimum 
volume of a cut must be at least half the minimum degree of the graph, 
which we normalize to $1$, and $ii)$ in computing $h_G$, we restrict to 
subsets of volume at most $\vol(G)$, which is at most $\Delta n$.} 
thus $T_2 = \bigO(\Delta^2 n^2)$.
\end{remark}

\begin{remark}[The extent of non-regularity]\label{rmk2}
Notice that the condition $k > 2 \Delta (\nmax/\nmin)$ implies
\[
	\max_{i \in [k]}\{\vol(V_i)\} 
	\leq \Delta \nmax 
	< \frac{k}{2} \nmin
  \leq \frac{k}{2} \min_{i \in [k]} \{\vol(V_i)\}.
\]
In other words, the \averaging{} dynamics gives a good community-sensitive
labeling when the communities are not too unbalanced in terms of their volumes.
Moreover, the smaller the number of communities the more the volume-balance
requirement is tight.
\end{remark}

In the remainder of this section, we first introduce further notation and then
state the main technical lemmas
(\cref{lem:length_projection,lem:constant_prob_different_colors,lem:sign_difference}),
that will be used in the proof of \cref{thm:main}, which concludes this
section.

Let $G = (V,E,w)$ be a clustered $k$-volume regular graph and, without loss of
generality, let $V_1, \dots, V_k$ be an arbitrary ordering of its communities.
We introduce a family of stepwise vectors that generalize Fiedler
vector~\cite{fiedler1989laplacian}, namely
\begin{equation}\label{eq:fiedler_vectors}
\left\{
\bm{\chi}_i 
= {\sqrt{{ \frac{\hat{m}_i}{m_i}}}} \bm{1}_{V_i} 
- {\sqrt{{\frac{m_i}{\hat{m}_i}}}} \bm{1}_{\Vunion_i}
\; : \; i \in [k-1]
\right\},
\end{equation}
where $\bone_{V_i}$ is the indicator vector of the set $V_i$
and, for convenience sake, we denoted by $m_i$  the volume of the $i$-th community, 
$\Vunion_i$ the set of all nodes in communities $i+1, \dots, k$, and $\hat{m}_i$
the volume of $\Vunion_i$, i.e.,
$m_i := \sum_{u \in V_i} \dv(u)$, $\Vunion_i := \bigcup_{h=i+1}^{k}V_h$, 
and $\hat{m}_i := \sum_{h=i+1}^{k} m_h.$
Note that vectors $\bm{\chi}_i$s are ``stepwise'' with respect to the 
communities of $G$ (i.e., for every $i \in [k-1]$, 
$\bm{\chi}_i(u) = \bm{\chi}_i(v)$ whenever $u$ and $v$ belong to the same community). 

Recall from \cref{eq:state_decomp} that the initial state vector 
can be written as $\bm{x} = \sum_{i=1}^{n}\alpha_i\bm{v}_i$.
Let $\bm{z} := \sum_{i=1}^{k} \alpha_i \bm{v}_i$ and note that 
$\bm{z} = \alpha_1\bone{} + \sum_{i=1}^{k-1} \gamma_i \bm{\chi}_i$
by applying \cref{lem:span_stepwise} and because 
$Span\left( \{\bone{}, \bm{\chi}_1, \dots, \bm{\chi}_{k-1} \}\right) 
= Span\left( \{\bone_{V_1}, \dots, \bone_{V_k} \} \right)$.
Let us now define the vector $\bm{y} := \bm{z} - \alpha_1 \bone{}$ 
or, equivalently,
\begin{equation}\label{eq:vector_y}
	\bm{y} = \sum_{i=1}^{k-1} \gamma_i \bm{\chi}_i,
	\text{ where }
	\gamma_i = \frac{\bm{x}^\T D \bm{\chi}_i}{\Norm{D^{1/2}  \bm{\chi}_i}^2}.
\end{equation}
Note that  the coefficients $\gamma_i$s
are proportional to the length of the projection of the (inhomogeneously)
contracted state vector on the (inhomogeneously) contracted
$\Dsq{}\bm{\chi}_i$s; 
the previous expression is valid since the vectors in $\{\Dsq{}\bone{}\}
\cup \{\Dsq{}\bm{\chi}_i : i \in [k-1]\}$ are mutually orthogonal.%
\footnote{The mutual orthogonality of the vectors, including $\Dsq{}\bone{}$,
is also one of the reasons why other ``simpler'' families of stepwise vectors,
e.g., the indicator vectors of the communities, are not used instead.}

In~\cref{lem:length_projection} we show that every component of $\bm{y}$, i.e.,
the projection of the (inhomogeneously) contracted initial state vector
$\Dsq\bm{x}$ on the (inhomogeneously) contracted vectors $\Dsq{}\bm{\chi}_i$s,
is not too small, w.h.p.

\begin{lemma}[Length of the projection of the state vector]\label{lem:length_projection}
	Let $G = (V, E, w)$ be a connected clustered $k$-volume-regular graph with $n$
	nodes and $k$-partitions $\mathcal{V} = \{V_1,\ldots,V_k\}$. Under 
	the hypotheses of Theorem \ref{thm:main}, for every $u \in V$,
	\[
	\Pr\left(\Abs{\bm{y}(u)} \geq \frac{1}{\Delta n} \right) 
	\geq 1 - \bigO\left(\frac{1}{\sqrt{n}}\right).
	\]
\end{lemma}
\begin{proof}
	Without loss of generality, we assume $u\in V_1$, which possibly 
	just amounts to a relabeling of the nodes. With this assumption, we 
	have 
	\[
		\bm{y}(u) = \gamma_1\bm{\chi}_1(u) = 
		\gamma_1\sqrt{\frac{\hat{m}_1}{m_1}},
	\]
	where the second equality follows from the definitions of the 
	$\bm{\chi}_i$'s (\cref{eq:fiedler_vectors}) and the fact that $u\in V_1$. 
	Next, observe that we 
	have: 
	\[
		\Norm{D^{1/2}  \bm{\chi}_1}^2 
		= {\frac{\hat{m}_1}{m_1}}\sum_{v \in V_1} \dv(v) + { 
		\frac{m_1}{\hat{m}_1}}\sum_{v \in \Vunion_1} \dv(v) 
		= \hat{m}_1 + m_1 = m,
	\]
	where $m := \vol(V)$. We now bound 
	\[
		\Abs{\gamma_1} = \frac{\Abs{\bm{x}^\T D 
		\bm{\chi}_1}}{\Norm{D^{1/2}  \bm{\chi}_1}^2} = \frac{\Abs{\bm{x}^\T D 
		\bm{\chi}_1}}{m}.
	\]
	More precisely, we prove that it is at least $1/m$ with probability 
	$1 - \bigO\left(\frac{1}{\sqrt{n}}\right)$, where probability is computed over 
	the randomness of $\bm{x}$. 
	
	Assume for the moment that $\hat{m}_1\ge m_1$. 
	From the definition of $\bm{\chi}_1$ we have:
	\begin{align*}
		&\bm{x}^\T D\bm{\chi}_1 = \bm{x}^\T D\left(\sqrt{{ 
		\frac{\hat{m}_1}{m_1}}} \bm{1}_{V_1} 
		- \sqrt{\frac{m_1}{\hat{m}_1}} \bm{1}_{\Vunion_1}\right) = 
		\sqrt{\frac{m_1}{\hat{m}_1}}\bm{x}^\T 
		D\left(\frac{\hat{m}_1}{m_1}\bm{1}_{V_1} - \bm{1}_{\Vunion_1}\right).
	\end{align*}
	Now, set $\bw = D\left(\frac{\hat{m}_1}{m_1}\bm{1}_{V_1} - 
	\bm{1}_{\Vunion_1}\right)$ and note 
	that $|\bw(u)|\ge 1$ from the hypothesis that $\hat{m}_1\ge m_1$ 
	and since $\dv(v)\ge 1$, for every $v\in V$. We can thus apply 
	Theorem \ref{thm:littlewood-offord} to $\bw$ with $r = 0$, so that we can write:
	\[
		\Prob{}{\Abs{\bm{x}^\T D 
		\bm{\chi}_1}<\sqrt{\frac{m_1}{\hat{m}_1}}} 
		= \Prob{}{|\bm{x}^T\bw| < 1}\le 
		\bigO\left(\frac{1}{\sqrt{n}}\right),
	\] 
	where the equality follows since $\bm{x}^\T D \bm{\chi}_1 = 
	\sqrt{\frac{m_1}{\hat{m}_1}}\bm{x}^T\bw$. Hence, with probability 
	$1 - \bigO\left(\frac{1}{\sqrt{n}}\right)$ we have 
	$\Abs{\gamma_1}\ge\sqrt{\frac{m_1}{\hat{m}_1}}\cdot\frac{1}{m}$ and thus, 
	with the same probability:
	\[
		\Abs{\bm{y}(u)} = 
		\Abs{\gamma_1}\sqrt{\frac{\hat{m}_1}{m_1}}\ge\frac{1}{m}\ge\frac{1}{\Delta 
		n}.
	\]

	Assume now that $m_1 > \hat{m}_1$. This time we write:
	\begin{align*}
		&\bm{x}^\T D\bm{\chi}_1 = 
		\sqrt{\frac{\hat{m}_1}{m_1}}\bm{x}^\T 
		D\left(\bm{1}_{V_1} - \frac{m_1}{\hat{m}_1}\bm{1}_{\Vunion_1}\right) 
	\end{align*}
	and we set $\bw = D\left(\bm{1}_{V_1} - 
	\frac{m_1}{\hat{m}_1}\bm{1}_{\Vunion_1}\right)$. Note that, again, 
	$|\bw(v)|\ge 1$ for every $v\in V$. Proceeding as in the previous 
	case we obtain 
	$\Abs{\gamma_1}\ge\sqrt{\frac{\hat{m}_1}{m_1}}\cdot\frac{1}{m}$ with 
	probability $1 - \bigO\left(\frac{1}{\sqrt{n}}\right)$ and thus, 
	with the same probability:
	\[
		\Abs{\bm{y}(u)} = 
		\Abs{\gamma_1}\sqrt{\frac{\hat{m}_1}{m_1}}
		\ge \frac{\hat{m}_1}{m m_1}
		= \frac{m - m_1}{m m_1}
		\stackrel{(a)}{\ge} \frac{1}{m}
		\ge \frac{1}{\Delta n},
	\]
	where in $(a)$ we used that $m_i < \frac{m}{2}$ (see \cref{rmk2}).
	This concludes the proof.
\end{proof}

In \cref{lem:constant_prob_different_colors} we show that given any
``pair of steps'' of the vector $\bm{y}$ (defined in \cref{eq:vector_y}), 
the two steps have different signs, with constant probability.

\begin{lemma}[Different communities, different signs]\label{lem:constant_prob_different_colors}
Let $G=(V,E,\w)$ be a clustered $k$-volume regular graph with maximum weighted
degree $\Delta\leq\frac{\sqrt{\nmin}}{25}$ and with $k > 2 \Delta (\nmax / \nmin)$. For each
pair of nodes $u \in V_i$ and $v \in V_j$, with $i \ne j$, it holds that
\[
	\Pr{ \,\sgn(\bm{y}(u)) \neq \sgn(\bm{y}(v))\, } = \Omega(1).
\]
\end{lemma}
\begin{proof}
	Since the ordering of the communities (and consequent definition of the
	$\bm{\chi}_i$'s, given in \cref{eq:fiedler_vectors}) is completely arbitrary, 
	we can assume $i = 1$ and $j = 2$, without loss of generality.	Let us define
    \(
		X(V_i) := \sum_{w \in V_i} \dv(w) \bm{x}(w)
	\),
	where $\bm{x}=\bm{x}^{(0)}$ is the initial state vector.
	
	Note that $\bm{y}(u) = \gamma_1 \bm{\chi}_1(u)$ and 
	$\bm{y}(v) = \gamma_1 \bm{\chi}_1(v) + \gamma_2 \bm{\chi}_2(v)$, 
	since the other terms of the $\bm{\chi}_i$s are equal to 0 on the components
	relative to $u$ and $v$.
	Thus, with some algebra, we get 
	\begin{align*}
		\bm{y}(u) &= \frac{1}{m} \left[ 
		\frac{\hat{m}_1}{m_1} X(V_1) - X(V_2) - X(\hat{V}_2)
		\right],
		\\
		\bm{y}(v) &= \frac{1}{m} \left[
		\frac{m_1 m_2 + m \hat{m}_2}{\hat{m}_1 m_2} X(V_2) - X(V_1) - X(\hat{V}_2)
		\right],
	\end{align*}
	where $\Vunion_i := \bigcup_{h=i+1}^{k}V_h$.
	Note that, by linearity of expectation, $\Ex{X(V_i)} = 0$.
	Moreover, since the terms $\bm{x}(w)$s are independent 
	Rademacher random variables, we can write the standard deviation of $X(V_i)$ as
	\[
	\sigma(X(V_i))
	= \sqrt{\sum_{w \in V_i} \sigma^2(\dv(w)\bm{x}(w))}
	= \sqrt{\sum_{w \in V_i} \left(\Ex{\dv(w)^2 \bm{x}(w)^2} - \Ex{\dv(w) \bm{x}(w)}^2\right)}
	= \sqrt{\sum_{w \in V_i} \dv(w)^2}.
	\]
	Then we can upper and lower bound the standard deviation $\sigma(X(V_i))$
	getting
	\(
		\frac{m_i}{\sqrt{n_i}}
		\leq \sigma(X(V_i)) 
		\leq \Delta \sqrt{n_i},
	\)
	where the lower bound follows from 
	$\Norm{\bm{d}}_2 \geq \Norm{\bm{d}}_1 / \sqrt{n_i}$, 
	where $\bm{d}_i$ is the vector of weighted degrees of nodes in community $V_i$,
	and for the upper bound we used that $\dv(w) \leq \Delta$, for each $w \in V$.

	Let us now define the following three events:
	\begin{enumerate}
		\item $E_1\text{: } X(V_1) > \sigma(X(V_1)) 
			\implies X(V_1) > \frac{m_1}{\sqrt{n_1}} 
				\geq \frac{\min_i m_i}{\sqrt{\nmax}}$;
		\item $E_2\text{: } X(V_2) < -\sigma(X(V_2))
			\implies X(V_2) < -\frac{m_2}{\sqrt{n_2}} 
				\leq -\frac{\min_i m_i}{\sqrt{\nmax}}$;
		\item $E_3\text{: } 0 \leq X(\hat{V}_2) < \left(2/\sqrt{k}\right) \sigma(X(\hat{V}_2))
			\implies 0 \leq X(\hat{V}_2) < 2 \Delta \sqrt{(1/k) \sum_{i=3}^k n_i}
				\leq 2 \Delta \sqrt{\nmax}$,
	\end{enumerate}
	When $E_1, E_2, E_3$ are true it follows directly that $\bm{y}(v) < 0$.
	As for $\bm{y}(u) > 0$ we have
	\begin{align*}
		\frac{\hat{m}_1}{m_1} X(V_1) - X(V_2) - X(\hat{V}_2)
		&> \frac{\hat{m}_1}{m_1} \sigma(X(V_1)) + \sigma(X(V_2)) - (2/\sqrt{k}) \cdot \sigma(X(\hat{V}_2))
		\\
		&\geq \frac{k \nmin}{\sqrt{\nmax}} - 2 \Delta \sqrt{\textstyle\nmax} > 0,
	\end{align*}
  since, for the last inequality, $k > 2 \Delta (\nmax/\nmin)$ by hypothesis.

	Note that all three events $E_1, E_2, E_3$ have probability at least constant 
	and, being the events independent, also $\Pr{E_1 \cap E_2 \cap E_3}$ is constant.
	Indeed, it is possible to prove the constant lower bounds on the probabilities by 
	approximating the random variables with Gaussian ones using Berry-Esseen's 
	theorem (\cref{thm:berry-esseen_non_iid}).
	Note that $X(V_1), X(V_2), X(\Vunion_2)$ all are of the form 
	$Z=\sum_{w \in T} Z_w$, for some $T \subseteq V$ and where $Z_w=\dv(w)\bm{x}(w)$. 
	Recall that $\Ex{Z_w} = 0$ and that $\sigma^2(Z_w) = \dv(w)^2$.
	Moreover, note that the third absolute moment of $Z_w$ is
	\(
		\Ex{\lvert Z_w \rvert^3} = \dv(w)^3 \Ex{\lvert\bm{x}(w)\rvert^3} = \dv(w)^3.
	\)
	Therefore we can apply \cref{thm:berry-esseen_non_iid} which claims that
	there exists a positive constant $C \leq 1.88$ \cite{berry1941accuracy} such that, 
	for every $z\in\mathbb{R}$,
	\[
		\left\lvert \Pr{Z\leq z \cdot \sigma(Z)} - \Phi(z) \right\rvert 
		\leq \frac{C}{\sigma(Z)} \max_{w \in T} \frac{\dv(w)^3}{\dv(w)^2}
		\leq \frac{C \cdot \Delta}{\sigma(Z)},
	\]
	where $\Phi$ is the cumulative distribution function of the standardized normal
	distribution. Thus
	\begin{equation}\label{eq:berry}
		\Pr{Z > z \cdot \sigma(Z)} 
		\geq 1 - \Phi(z) - \frac{C \cdot \Delta}{\sigma(Z)}.
	\end{equation}
	
	Since $\Delta \leq \frac{\sqrt{\nmin}}{25}$ by hypothesis
	and $\sigma(Z) \geq \sqrt{|T|}$ for every $T \subseteq V$, 
	taking $z=1$ it follows from \cref{eq:berry} that
	\[
		\Pr{E_1} = \Pr{X(V_1) > \sigma(X(V_1))}
		\geq 1 - \Phi(1) - \frac{C \cdot \Delta}{\sqrt{\nmin}}
		\geq 1 - \Phi(1) - \frac{C}{25}
		\geq \frac{1 - \Phi(1)}{2}
		\approx 0.08,
	\]
	since $\frac{1}{25} < \frac{1-\Phi(1)}{2C} \approx 0.042$.
	Since the distribution of $Z$ is symmetric for every $T \subseteq V$, 
	it holds that $\Pr{E_2} \geq 1 - \Phi(1) - \frac{C}{25} \approx 0.08$.
	Similarly, it also holds that 
	\[
		\Pr{E_3} = \frac{1}{2} - \Pr{X(\Vunion_2) > \left(2 / \sqrt{k}\right) \sigma(X(\Vunion_2))}
		\geq\Phi(2/\sqrt{k}) + \frac{C}{25}-\frac{1}{2} > \frac{C}{25} \approx 0.075.
		\qedhere
	\]
\end{proof}

Recall that the binary labeling of each node only depends on the difference of
its state in two consecutive rounds (see \cref{alg:avg_dyn}).
In~\cref{lem:sign_difference} we show that, under suitable 
assumptions on the transition matrix of a random walk on $G$, a large 
enough time window exists where, for each node $u$, the sign of the 
difference $\bm{x}^{(t)}(u) - \bm{x}^{(t+1)}(u)$ of the state 
vector across two consecutive rounds equals 
the sign of $\bm{y}(u)$, w.h.p. Since $\bm{y}$ (defined in \cref{eq:vector_y})
is a stepwise vector, this implies that two nodes in the same community have the 
same label, w.h.p.
For the sake of readability, in the proof of \cref{lem:sign_difference} we use 
two technical lemmas as black boxes, postponing their proofs to
Subsection~\ref{sec:techproof}.

\begin{lemma}[Sign of the difference]\label{lem:sign_difference}
Let $G = (V,E,\w)$ be a clustered $k$-volume regular graph 
with maximum weighted degree $\Delta\leq\frac{\sqrt{\nmin}}{25}$.
If $\lambda_k \geq \frac{7\lambda_2-5}{2}$,
$1-\lambda_2 \le \frac{\lambda_k \log(\lambda_k/\lambda_{k+1})}{7\log(2\Delta n)}$,
$\Abs{\bm{y}(u)} \geq \frac{1}{\Delta n}$ for every $u \in V$,
then a non-empty time interval $[T_1, T_2]$ exists,
with $T_1 = \bigO\left(\frac{\log n}{\log(\lambda_k / \lambda_{k+1})}\right)$ 
and $T_2 = \Omega\left(\frac{\lambda_k}{1-\lambda_2}\right)$,
such that, for every $u \in V$ and every $t \in [T_1, T_2]$ 
of the \averaging{} dynamics,
\[
	\sgn (\, \bm{x}^{(t)}(u) - \bm{x}^{(t+1)}(u) \,) = \sgn ( \bm{y}(u) ).
\]
\end{lemma}
\begin{proof}
Recall from \cref{eq:state_decomp} that the state vector at time $t$,
i.e., $\bm{x}^{(t)}$, can be written as the sum of the first $k$ stepwise 
vectors of $P$ and of the remaining ones, namely
\[
	\bm{x}^{(t)}
	= \alpha_1 \bm{1} + \sum_{i=2}^k \lambda_i^t \alpha_i \bm{v}_i 
		+ \sum_{i=k+1}^n \lambda_i^t \alpha_i \bm{v}_i
	= \alpha_1 \bm{1} + \bm{c}^{(t)} + \xerr^{(t)}\,.
\]
In what follows we call $\bm{c}^{(t)} := \sum_{i=2}^k \lambda_i^t \alpha_i
\bm{v}_i$ the \emph{core contribution} and $\xerr^{(t)} := \sum_{i=k+1}^n
\lambda_i^t \alpha_i \bm{v}_i$ the \emph{error contribution}.
If we look at the difference of the state vector in two consecutive rounds,
the first term cancels out being constant over time, so that
\[
	\bm{x}^{(t)}(u) - \bm{x}^{(t+1)}(u)
	= \bm{c}^{(t)}(u) - \bm{c}^{(t+1)}(u) + \xerr^{(t)}(u) - \xerr^{(t+1)}(u)
\]
for each node $u \in V$.
Note that the sign of the difference between two consecutive rounds is
determined by the difference of the core contributions, $\bm{c}^{(t)}(u) -
\bm{c}^{(t+1)}(u)$, whenever
\begin{equation}\label{eq:whenever}
	\Abs{\bm{c}^{(t)}(u) - \bm{c}^{(t+1)}(u)}
	> \Abs{\xerr^{(t)}(u) - \xerr^{(t+1)}(u)}.
\end{equation}
To identify conditions on $t$ for which \cref{eq:whenever} holds,
we give suitable bounds on both hand sides of the inequality.
In more detail: 
\begin{enumerate}
\item In~\cref{lemma:norm_of_e} we prove that
\(
	\abs{\xerr^{(t)}(u)} \leq \lambda_{k+1}^t \sqrt{\Delta n},
\)
for every $u \in V$, so that
\[
	\Abs{\xerr^{(t)}(u) - \xerr^{(t+1)}(u)} 
	\leq \Abs{\xerr^{(t)}(u)} + \Abs{\xerr^{(t+1)}(u)} 
	\leq 2 \lambda_{k+1}^{t} \sqrt{\Delta n}.
\]

\item In \cref{lemma:lower_bound_diff_c} we prove that 
\(
	\Abs{\bm{c}^{(t)}(u) - \bm{c}^{(t+1)}(u)}
	> \lambda_k^t (1-\lambda_2) \Abs{ \bm{y}(u) }
\)
for every $u \in V$ and for every time $t < T_2$, 
where $T_2 \geq \frac{\lambda_k}{2(1-\lambda_2)}$;
note that the hypotheses on $1-\lambda_2$ imply 
$T_2 = \Omega\left(\frac{\log n}{\log(\lambda_k/\lambda_{k+1})}\right)$.
Moreover, the assumptions of \cref{lemma:lower_bound_diff_c} are 
satisfied, since $\bm{y}(u) \neq 0$ 
and $\lambda_k \geq \frac{7\lambda_2-5}{2}$.
\end{enumerate}

\noindent Combining \cref{lemma:norm_of_e} and \cref{lemma:lower_bound_diff_c},
we see that \cref{eq:whenever} holds whenever 
\begin{equation}\label{eq:time}
	\lambda_k^t (1-\lambda_2) \Abs{ \bm{y}(u) }
	> 2 \lambda_{k+1}^{t} \sqrt{\Delta n},
\end{equation}
An easy calculation shows that this happens for all $t > T_1$, where 
\[
	T_1 := \frac{\log\left( \frac{2 \sqrt{\Delta n}}{(1-\lambda_2) \Abs{ \bm{y}(u) }}\right)}{\log \left( \lambda_k / \lambda_{k+1} \right)}.
\]
Note that $T_1 = \bigO\left(\frac{\log n}{\log(\lambda_k / \lambda_{k+1})}\right)$ and, e.g.,
$T_1 = \bigO(\log n)$ whenever $\frac{\lambda_k}{\lambda_{k+1}} = 1 + \Omega(1)$.

We next show that, under the assumptions of the lemma, the window $[T_1, T_2]$ is not empty and, actually, 
it has a width that depends on the magnitude of $\lambda_2$ and the 
ratio $\lambda_k/\lambda_{k+1}$. To this purpose, we first observe that 
Cheeger's inequality for weighted graphs (\cref{thm:cheeger}) implies
$1-\lambda_2 \geq \frac{h_G^2}{2} \geq \frac{1}{2 (\Delta 
n)^2}$ (recall the footnote in Remark \ref{rmk1}). Moreover, 
recalling that we are assuming $\Abs{\bm{y}(u)} \geq 
\frac{1}{\Delta n}$,%
\footnote{It may be worth recalling that our hypothesis on $\Abs{\bm{y}(u)}$ 
holds with high probability from \cref{lem:length_projection}.} 
we have:

\begin{equation}\label{eq:time_window_not_empty}
	T_1 = \frac{\log\left( \frac{2 \sqrt{\Delta n}}{(1-\lambda_2) \Abs{ \bm{y}(u) }}\right)}{\log \left( \lambda_k / \lambda_{k+1} \right)}
	\stackrel{(a)}{\leq} \frac{\log\left(4 \Delta^3 n^3 \sqrt{\Delta n}\right)}{\log \left( \lambda_k / \lambda_{k+1} \right)}
	< \frac{7\log(2\Delta n)}{2\log \left( \lambda_k / \lambda_{k+1} \right)}
	\stackrel{(b)}{\leq} \frac{\lambda_k}{2(1-\lambda_2)}
	\stackrel{(c)}{\leq} T_2,
\end{equation}
where: in $(a)$ we used Cheeger's inequality (\cref{thm:cheeger}) 
in the way described above and our assumptions on $\Abs{\bm{y}(u)}$,
in $(b)$ we used the hypothesis on $1-\lambda_2$, which implies
$\log(\lambda_k/\lambda_{k+1}) \ge \frac{7 (1-\lambda_2)\log(2\Delta n)}{\lambda_k}$,
and in $(c)$ the lower bound on $T_2$ given by \cref{lemma:lower_bound_diff_c}.

From \cref{lemma:lower_bound_diff_c} we also know that 
$= \sgn ( \bm{c}^{(t)}(u) - \bm{c}^{(t+1)}(u) ) = \sgn ( \bm{y}(u) )$
for every time $t < T_2$; therefore we conclude that
\[
\sgn ( \bm{x}^{(t)}(u) - \bm{x}^{(t+1)}(u) ) 
= \sgn ( \bm{c}^{(t)}(u) - \bm{c}^{(t+1)}(u) )
= \sgn ( \bm{y}(u) )\,,
\]
for every node $u \in V$ and for every round $t \in [T_1, T_2]$ of the
\averaging{} dynamics.
\end{proof}

\noindent
\begin{proof}[Proof of \cref{thm:main}]
The binary labeling of the nodes of $G$ produced by the \averaging{} dynamics
during the time window $[T_1, T_2]$ is such that the two conditions required by
the definition of $(\varepsilon, \delta)$-community-sensitive algorithm
(\cref{def:comm_sen_alg}) are met, with $\varepsilon = \bigO(n^{-1/2})$ and
$\delta = 1 - \Omega(1)$.  Indeed, the first condition follows directly from
\cref{lem:sign_difference} together with the fact that $\bm{y}$ is a ``stepwise''
vector, while \cref{lem:length_projection} implies that $\varepsilon =
\bigO(n^{-\frac{1}{2}})$, since $\bm{y}(u)$ is not too small with probability at
least $1 - \bigO(n^{-1/2})$. The second condition, instead, follows directly 
from the combination of \cref{lem:sign_difference,lem:constant_prob_different_colors}.
\end{proof}

\begin{remark}[Equal-sized communities]\label{rmk:equa_sized}
	If $\lambda_k = \lambda_2$, then an alternative version of 
	\cref{lemma:lower_bound_diff_c} would tell us that, for every node $u \in V$,
	$\bm{c}^{(t)}(u)-\bm{c}^{(t+1)}(u) = \lambda_k^{t}(1-\lambda_2)\bm{y}(u)$
	and thus $\sgn(\bm{c}^{(t)}(u)-\bm{c}^{(t+1)}(u)) = \sgn(\bm{y}(u))$,
	in every round (with no need of $T_2$);
	this would imply an infinite time window starting at the first round $t > T_1$
	(where the ``error contribution'' becomes small). 
	In this sense our result also covers the case of multiple communities
	analyzed in~\cite{becchetti2017find}, with $k$ equal-sized communities in an 
	unweighted graph and then $\lambda_k = \lambda_2$.
\end{remark}

\begin{remark}[Two communities]\label{rmk:two_com}
	Our result also generalizes that of~\cite{becchetti2017find} in the simpler 
	case of two communities. In fact we don't require the graph to be regular, 
	but only volume-regular, thus taking into account communities that are 
	potentially unbalanced.
	Ideed, for $k = 2$, the \avg{} dynamics truncated at round $t$ 
	is a $(\bigO(n^{-\frac{1}{2}}),\,\bigO(n^{-\frac{1}{2}}))$-community-sensitive 
	algorithm for every round $t > T_1$, with 
	$T_1 = \bigO\left( \frac{\log n}{\log(\lambda_2/\lambda_3)} \right)$.
	Therefore, a single run of the dynamics highlights the community structure, i.e.,
	the sign of the difference $\bm{x}^{(t)}-\bm{x}^{(t)}$ is equal for nodes in 
	the same communities and different for nodes in different communities, 
	w.h.p.
\end{remark}

\subsection{Proofs for Lemma~\ref{lem:sign_difference}}\label{sec:techproof}
In this section we prove the two lemmas used in the proof of
\cref{lem:sign_difference}: the upper bound on the ``error contribution'' and
the lower bound on the ``core contribution.''

\begin{lemma}[Upper bound on the error contribution]\label{lemma:norm_of_e}
Let $\xerr^{(t)} := \sum_{i=k+1}^{n} \lambda_i^t \alpha_i \bm{v}_i$.
For every $u \in V$, it holds that
\[
	\abs{\xerr^{(t)}(u)} \leq \lambda_{k+1}^t \sqrt{\Delta n}.
\]
\end{lemma}
\begin{proof}
	To bound all components of vector $\xerr^{(t)}$ we use its $\ell^\infty$ norm,
	defined for any vector $\bm{x}$ as $\norm{ \bm{x} }_\infty := \sup_i \abs{\bm{x}(i)}$.
	In particular
	\begin{align*}
		\norm{\xerr^{(t)}}_\infty^2
		&\leq \norm{\xerr^{(t)}}^2
		= \Norm{ \sum_{i=k+1}^n \lambda_i^t \alpha_i \bm{v}_i }^2
		= \Norm{ \sum_{i=k+1}^n \lambda_i^t \beta_i \Dsqinv \bm{w}_i }^2.
	\end{align*}
	By using Cauchy-Schwarz inequality (\cref{thm:cauchy-schwarz}) 
	and applying the definition of spectral norm of an operator, 
	i.e., $\norm{A} := \sup_{\bm{x}:\norm{\bm{x}=1}} \norm{A\bm{x}}$,
	we get that
	\[
		\Norm{ \sum_{i=k+1}^n \lambda_i^t \beta_i \Dsqinv \bm{w}_i }^2
		\leq \Norm{\Dsqinv}^2 \Norm{ \sum_{i=k+1}^n \lambda_i^t \beta_i \bm{w}_i }^2
		= \Norm{\Dsqinv}^2 \sum_{i=k+1}^n \lambda_i^{2t} \beta_i^2,
	\]
	since the $\bm{w}_i$s are orthonormal. 
	With some additional simple bounds it follows that
	\begin{align*}
		\Norm{\Dsqinv}^2 \sum_{i=k+1}^n \lambda_i^{2t} \beta_i^2
		&\leq \Norm{\Dsqinv}^2 \lambda_{k+1}^{2t} \sum_{i=k+1}^n \beta_i^2
		\leq \Norm{\Dsqinv}^2 \lambda_{k+1}^{2t} \sum_{i=1}^n \beta_i^2
		\\
		&= \Norm{\Dsqinv}^2 \lambda_{k+1}^{2t} \Norm{\Dsq \bm{x}}^2
		\leq \Norm{\Dsqinv}^2 \lambda_{k+1}^{2t} \Norm{\Dsq}^2 \Norm{\bm{x}}^2.
	\end{align*}
	By using the fact that the spectral norm of a diagonal matrix 
	is equal to its maximum value, we conclude that
	\[
		\Norm{\Dsqinv}^2 \lambda_{k+1}^{2t} \Norm{\Dsq}^2 \Norm{\bm{x}}^2
		= \frac{\max_u{\dv(u)}}{\min_u{\dv(u)}} \lambda_{k+1}^{2t} \Norm{ \bm{x} }^2
		\leq \lambda_{k+1}^{2t} \Delta n.
	\]
	Thus, for every $u \in V$ it holds that
	\(
	\abs{\xerr^{(t)}(u)} 
	\leq \sqrt{\norm{\xerr^{(t)}}_\infty^2}
	\leq \lambda_{k+1}^t \sqrt{\Delta n}.
	\)
\end{proof}
%


\smallskip\noindent In \cref{lemma:lower_bound_diff_c} we show that the difference
of the core contribution in consecutive rounds can be approximated, for our 
purposes in \cref{lem:sign_difference}, with $\bm{y}$.

\begin{lemma}[Lower bound on the core contribution]\label{lemma:lower_bound_diff_c}
	Let $\bm{c}^{(t)} := \sum_{i=2}^{k} \lambda_i^t \alpha_i \bm{v}_i$ 
	and let $\bm{y}(u) \neq 0$ for every $u \in V$.
	If $\lambda_k \geq \frac{7\lambda_2-5}{2}$,
	then, for every $u \in V$ and for every $t\leq T_2$, 
	with $T_2 \geq \frac{\lambda_k}{2(1-\lambda_2)}$, the following 
	holds:
	\begin{itemize}
		\item $\sgn(\bm{c}^{(t)}(u) - \bm{c}^{(t+1)}(u)) = \sgn( \bm{y}(u) )$;
		\item $\Abs{ \bm{c}^{(t)}(u) - \bm{c}^{(t+1)}(u) }
		> \lambda_k^{t} (1-\lambda_2) \Abs{ \bm{y}(u) }$.	
	\end{itemize}
\end{lemma}
\begin{proof}
	Let us define $d_{i,j} := \lambda_i - \lambda_j$. Note that
	\begin{align*}
	\bm{c}^{(t)} - \bm{c}^{(t+1)} &= \sum_{i=2}^{k} \lambda_i^t (1 - \lambda_i) \alpha_i \bm{v}_i
	= \sum_{i=2}^{k} (\lambda_k + d_{i,k})^t (1 - \lambda_2 + d_{2,i}) \alpha_i \bm{v}_i
	\\
	&= \sum_{i=2}^{k} [\lambda_k^t + (\lambda_k + d_{i,k})^t - \lambda_k^t] (1 - \lambda_2 + d_{2,i}) \alpha_i \bm{v}_i
	= \sum_{i=2}^{k} [\lambda_k^t (1 - \lambda_2) + c_i] \alpha_i \bm{v}_i
	\\
	&= \lambda_k^t (1 - \lambda_2)\sum_{i=2}^{k} \alpha_i \bm{v}_i + \sum_{i=2}^{k} c_i \alpha_i \bm{v}_i
	= \lambda_k^t (1 - \lambda_2)\bm{y} + \sum_{i=2}^{k} c_i \alpha_i \bm{v}_i,
	\end{align*}
	where in the last equality we applied \cref{lem:span_stepwise} to get
	$\sum_{i=2}^{k} \alpha_i \bm{v}_i = \bm{y}$, and where we defined 
	$c_i := \lambda_k^t d_{2,i} + [(\lambda_k+d_{i,k})^t-\lambda_k^t] (1-\lambda_2+d_{2,i})$.
	Using the definition of $d_{i,j}$, we get
	\begin{align*}
	c_i &= \lambda_k^t d_{2,i} + [(\lambda_k+d_{i,k})^t-\lambda_k^t] (1-\lambda_2+d_{2,i})
	\\
	&= \lambda_k^t (\lambda_2-\lambda_i) + [(\lambda_k+(\lambda_i-\lambda_k))^t-\lambda_k^t] (1-\lambda_2+(\lambda_2-\lambda_i))
	\\
	&= \lambda_k^t (\lambda_2-\lambda_i) + (\lambda_i^t-\lambda_k^t) (1-\lambda_i)
	\\
	&= \lambda_i^t (1-\lambda_i) - \lambda_k^t(1-\lambda_2).
	\end{align*}
	Note that 
	$\min_i[ \lambda_i^t (1-\lambda_i) ] - \lambda_k^t(1-\lambda_2)
	\leq c_i \leq 
	\max_i [ \lambda_i^t (1-\lambda_i) ] - \lambda_k^t(1-\lambda_2)$, 
	for every $i \in \{2,\dots,k\}$. 
	Since the minimum and the maximum are obtained for $i=k$ and $i=2$ respectively,
	we have
	$\lambda_k^t(\lambda_2-\lambda_k) 
	\leq c_i \leq 
	(\lambda_2^t - \lambda_k^t)(1-\lambda_2)$.
	Let us call the positive and negative terms of $\bm{y}(u)$ as
	\[
		\bm{y}^+(u) := \sum_{\Pos}\alpha_i\bm{v}_i(u)
		\quad\qquad\text{and}\quad\qquad
		\bm{y}^-(u) := -\sum_{\Neg}\alpha_i\bm{v}_i(u).
	\]
	Therefore, for each $u \in V$, it holds that
	\begin{align}\label{eq:lb_diff_c}
		\bm{c}^{(t)}(u) - \bm{c}^{(t+1)}(u) &\geq \lambda_k^t (1 - \lambda_2)\bm{y}(u)
		+ \lambda_k^t(\lambda_2 - \lambda_k) \bm{y}^+(u)
		- (\lambda_2^t - \lambda_k^t)(1-\lambda_2) \bm{y}^-(u),
	\\
	\label{eq:ub_diff_c}
		\bm{c}^{(t)}(u) - \bm{c}^{(t+1)}(u) &\leq \lambda_k^t (1 - \lambda_2)\bm{y}(u) 
		+ (\lambda_2^t - \lambda_k^t)(1-\lambda_2) \bm{y}^+(u)
		- \lambda_k^t(\lambda_2 - \lambda_k) \bm{y}^-(u).
	\end{align}

	\bigskip
	In the following we look for a time $T_2$ such that, for every $t \leq T_2$ 
	it holds that
	\[
		\Abs{ \bm{c}^{(t)}(u) - \bm{c}^{(t+1)}(u) } 
		> \lambda_k^{t} (1-\lambda_2) \Abs{ \bm{y}(u) }.
	\]

	\noindent
	Note that $\bm{y}(u)=\bm{y}^+(u)-\bm{y}^-(u)$. We consider two cases:
	$\bm{y}(u) > 0$ and $\bm{y}(u) < 0$.

	\smallskip\noindent \underline{Case $\bm{y}(u) > 0$}: 
	We look for a time $T_2$ such that, for every time $t \leq T_2$, it holds that
	\begin{equation}\label{eq:y_pos}
		\lambda_k^t(\lambda_2 - \lambda_k) \bm{y}^+(u) 
		> (\lambda_2^t - \lambda_k^t)(1-\lambda_2) \bm{y}^-(u).
	\end{equation}
	Indeed, since $\bm{y}(u) > 0$, we can use  
	$\bm{y}^+(u) > \bm{y}^-(u)$
	to upper bound the right hand side of the previous equation,
	so that \cref{eq:y_pos} holds for every $t$ that satisfies:
	\[
		\lambda_k^t(\lambda_2 - \lambda_k) > (\lambda_2^t - \lambda_k^t)(1-\lambda_2)
		\iff
		\left(\frac{\lambda_2}{\lambda_k}\right)^t < \frac{1-\lambda_k}{1-\lambda_2},
	\] 
	i.e., for every $t \leq T_2$, where
	\[
		T_2 := \frac{\ln\left(\frac{1-\lambda_k}{1-\lambda_2} \right)}{\ln\left(\frac{\lambda_2}{\lambda_k}\right)}
		= \frac{\ln\left(1+\frac{\lambda_2-\lambda_k}{1-\lambda_2} \right)}{\ln\left(1+\frac{\lambda_2-\lambda_k}{\lambda_k}\right)}.
	\]
	Next, note that $\frac{\lambda_2-\lambda_k}{1-\lambda_2} \leq \frac{5}{2}$ 
	whenever $\lambda_k \geq \frac{7\lambda_2-5}{2}$. 
	Hence, under the hypotheses of the lemma, we can use that 
	$1+x \geq e^{\frac{x}{2}}$ for every $x \in [0,\frac{5}{2}]$
	and $1+x \leq e^x$ for every $x$. Thus:
	\[
		T_2 = \frac{\ln\left(1+\frac{\lambda_2-\lambda_k}{1-\lambda_2} \right)}{\ln\left(1+\frac{\lambda_2-\lambda_k}{\lambda_k}\right)}
		\geq \frac{\left(\frac{\lambda_2-\lambda_k}{2(1-\lambda_2)}\right)}{\left(\frac{\lambda_2-\lambda_k}{\lambda_k}\right)}
		= \frac{\lambda_k}{2(1-\lambda_2)}.
	\]
	Plugging \cref{eq:y_pos} into \cref{eq:lb_diff_c} we finally get 
	\begin{equation}\label{eq:diff_c_pos}
		\bm{c}^{(t)}(u) - \bm{c}^{(t+1)}(u) > \lambda_k^t (1 - \lambda_2)\bm{y}(u) > 0.
	\end{equation}
	
	\smallskip\noindent \underline{Case $\bm{y}(u) < 0$}: 
	Proceeding along the same lines we obtain:
	\begin{equation}\label{eq:y_neg}
		(\lambda_2^t - \lambda_k^t)(1-\lambda_2) \bm{y}^+(u)
		< \lambda_k^t(\lambda_2 - \lambda_k) \bm{y}^-(u)
	\end{equation}
	for every $t\leq T_2$.
	Therefore, by combining \cref{eq:y_neg} and \cref{eq:ub_diff_c} we 
	obtain
	\begin{equation}\label{eq:diff_c_neg}
		\bm{c}^{(t)}(u) - \bm{c}^{(t+1)}(u) < \lambda_k^t (1 - \lambda_2)\bm{y}(u) < 0.
	\end{equation}

	\bigskip\noindent Finally, by combining \cref{eq:diff_c_pos} and 
	\cref{eq:diff_c_neg}:
	\begin{itemize}
		\item $\sgn(\bm{c}^{(t)}(u) - \bm{c}^{(t+1)}(u)) = \sgn( \bm{y}(u) )$;
		\item $\Abs{ \bm{c}^{(t)}(u) - \bm{c}^{(t+1)}(u) }
		> \lambda_k^{t} (1-\lambda_2) \Abs{ \bm{y}(u) }$.\qedhere
	\end{itemize}%
\end{proof}

\section{Bipartite graphs}\label{sec:bipartite}

Assume $G=(V,E,\w)$ is an edge-weighted bipartite graph with 
$V = V_1\cup V_2$ and $E\subseteq V_1\times V_2$, 
i.e. a graph with hidden partition identified by the bipartition.
In this case, basic properties of random walks imply that 
the \averaging{} dynamics does not converge to the global 
(weighted) average of the values, but it periodically oscillates. 
In fact, in this case the transition matrix $P$ has an eigenvector 
$\bm{\chi}= \bm{1}_{V_1} - \bm{1}_{V_2}$ with eigenvalue $\lambda_n = -1$ 
(as implied by \cref{lemma:sym_spec}).
Thus, the state vector is mainly affected by the eigenvectors associated to
the two eigenvalues of absolute value $1$ (i.e., $\lambda_1$ and $\lambda_{n}$).
After a number of rounds of the dynamics that depends on $1/\lambda_2$, 
we have that, in even rounds,
all nodes in $V_i$ ($i = 1, 2$) have a state that is close to some local average 
$\mu_i$; in odd rounds, these values are swapped
(as shown in \cref{eq:bip_spec_dec}).

If one were observing the process in even rounds,\footnote{Or, 
equivalently, in odd rounds.} however, 
the states of  nodes in $V_1$ would converge to $\mu_1$ 
and those of nodes in $V_2$ would converge to $\mu_2$.
Unfortunately, convergence to the local average for nodes belonging to the same 
community does not eventually become monotone (i.e., increasing or decreasing). 
This follows since the eigenvector associated to 
$\lambda_2$ is no longer stepwise in general.
However, we can easily modify the labeling scheme of the \avg{} 
dynamics to perform \emph{bipartiteness detection} as follows: 
Nodes apply the labeling rule every two time steps and they do it between the 
states of two consecutive rounds, i.e., 
each node $v \in V$ sets $\texttt{label}^{(2t)}(v) = 1$ 
if $\bx^{(2t)}(v) \geqslant \bx^{(2t-1)}(v)$ 
and $\texttt{label}^{(2t)}(v) = 0$ otherwise.
We call this new protocol \avgbip{} dynamics.

We now show how \avgbip{} dynamics can perform \emph{bipartiteness detection}.
Recall that we denote with $W \in\mathbb{R}^{n\times n}$ the weighted adjacency
matrix of $G$. Since $G$ is undirected and bipartite, the matrix $W$ can be written as
\[
W = \left( 
\begin{array}{cc}
0 & W_1\\
W_2 & 0
\end{array}\right) 
= 
\left(\begin{array}{cc}
0 & W_1\\
W_1^\T & 0
\end{array}\right).
\]
Thus, the transition matrix of a random walk on $G$, i.e., $P=D^{-1}W$
where $D^{-1}$ is a diagonal matrix and $D_{ii}=\frac{1}{\dv(i)}$, has the form 
\[
P = \left( 
\begin{array}{cc}
0 & P_1\\
P_1^\T & 0
\end{array}\right).
\]
\cref{lemma:sym_spec} shows that the spectrum of $P$ is symmetric and
it gives a relation between the eigenvectors of symmetric eigenvalues.
\begin{lemma}\label{lemma:sym_spec}
	Let $G=(V_1 \cup V_2,E,w)$ be an edge-weighted undirected bipartite graph
	with bipartition $(V_1,V_2)$ and such that $\abs{V_i}=n_i$.
	If $\bm{v} = (\bm{v}_1,\bm{v}_2)^T$, with $\bm{v}_i\in\mathbb{R}^{n_i}$, 
	is an eigenvector of $P$ with eigenvalue $\lambda$,
	then $\bm{v}' = (\bm{v}_1, -\bm{v}_2)^\T$ is an eigenvector of $P$ 
	with eigenvalue $-\lambda$.
\end{lemma}
\begin{proof}
	If $P \bm{v} = \lambda \bm{v}$ then we have that  
	$P_1 \bm{v}_2 = \lambda\bm{v}_1$ and $P_1^\T \bm{v}_2 = \lambda\bm{v}_2$. 
	Using these two equalities we get that $P\bm{v}'=-\lambda\bm{v}'$.
	Indeed,
	\[
	P\bm{v}' =
	\left(\begin{array}{cc}
	0 & P_1\\
	P_1^\T & 0
	\end{array}\right) \left(\begin{array}{c}
	\bm{v}_1 \\ -\bm{v}_2
	\end{array}\right) =
	\left(\begin{array}{c}
	-P_1 \bm{v}_2 \\ P_1^\T\bm{v}_1
	\end{array}\right) = 
	-\lambda \left(\begin{array}{c}
	\bm{v}_1 \\ -\bm{v}_2	
	\end{array}\right).%
	\qedhere
	\]
\end{proof}

The transition matrix $P$ is stochastic, thus the vector $\bm{1}$ 
(i.e., the vector of all ones) is an eigenvector associated 
to $\lambda_1 = 1$, that is the first largest eigenvalue of $P$. 
\cref{lemma:sym_spec} implies that $\bm{\chi} = \bm{1}_{V_1} - \bm{1}_{V_2}$ 
is an eigenvector of $P$ with eigenvalue $\lambda_{n}=-1$.

As in \cref{sec:preli}, we write the state vector at time $t$ using the spectral decomposition of $P$.
Let $1 = \lambda_1 > \lambda_2 \geq \ldots > \lambda_{n}= -1$ 
be the eigenvalues of $P$. 
We denote by $\bm{1} = \bm{v}_1, \bm{v}_2, \ldots, \bm{v}_{n}= \bm{\chi}$ 
a family of $n$ linearly independent eigenvectors of $P$,
where each $\bm{v}_i$ is the eigenvector associated to $\lambda_i$. 
Thus, we have that
\begin{equation}\label{eq:bip_spec_dec}
	\bm{x}^{(t)}= P^t\bm{x} = 
	\sum_{i=1}^{n}\lambda_i^t\alpha_i\bm{v}_i = 
	\alpha_1\bm{1} + (-1)^t\alpha_{n}\bm{\chi} + \sum_{i=2}^{n-1}\lambda_i^t\alpha_i\bm{v}_i 
\end{equation}
where $\alpha_i = \frac{\langle \Dsq \bm{x}, \Dsq \bm{v}_i\rangle}{\norm{\Dsq \bm{v}_i}^2}$. 
The last equation implies that $\bm{x}^{(t)}= P^t\bm{x}$ 
does not converge to some value as $t$ tends to infinity, but oscillates. 
In particular, nodes in $V_1$ on even rounds and
nodes in $V_2$ on odd rounds, converge to $\alpha_1 + \alpha_n$. 
Instead in the symmetric case, i.e., odd rounds for nodes in $V_1$ 
and even rounds for nodes in $V_2$, 
the process converges to $\alpha_1 - \alpha_n$. 
These quantities are proportional to the weighted average 
of the initial values in the first and in the second partition, respectively.

\cref{thm:bip_detection}, whose proof follows, shows that \avgbip{} dynamics
performs bipartiteness detection in $\bigO(\log n \,/ \,\log(1 / \lambda_2))$
rounds.
Note that, as in the case of volume-regular graphs with two communities (see \cref{rmk:two_com}), 
one single run of the dynamics identifies the bipartition. 
Moreover, if $\log(1 / \lambda_2) = \Omega(1)$, then the \avgbip{} dynamics takes
logarithmic time to find the bipartition.
\begin{theorem}\label{thm:bip_detection}
	Let $G=(V, E, \w)$ be an edge-weighted bipartite graph with bipartition $(V_1, V_2)$ 
	and maximum weighted degree $\Delta = \bigO(n^K)$, for any arbitrary positive constant $K$.
	Then for every time $t > T$, with $T = \bigO(\log n \,/ \,\log(1 / \lambda_2))$, 
	the \avgbip{} dynamics truncated at round $t$ is a 
	$(\bigO(n^{-\frac{1}{2}}),\, \bigO(n^{-\frac{1}{2}}))$-community-sensitive algorithm.
\end{theorem}

\begin{proof}
	We assume that the labeling rule is applied between every even and every odd
	round (conversely, the signs of the nodes in the analysis are swapped).
	Recall the definition of the \emph{error contribution}, namely
	\(
	\bm{e}^{(t)}(u) = \sum_{i=2}^{n-1}\lambda_i^t \alpha_i\bm{v}_i(u).
	\)
	We compute the difference between the state vectors of two consecutive steps 
	by using \cref{eq:bip_spec_dec}, namely
	\begin{align*}
		\bm{x}^{(2t)} - \bm{x}^{(2t+1)} 
		&= \alpha_1\bm{1} + (-1)^{2t}\alpha_{n}\bm{\chi} + \bm{e}^{(2t)}
		- \alpha_1\bm{1} - (-1)^{2t+1}\alpha_{n}\bm{\chi} - \bm{e}^{(2t+1)}
		\\
		& = 2\alpha_{n}\bm{\chi} + \bm{e}^{(2t)} - \bm{e}^{(2t + 1)}.
	\end{align*}
	We want to find a time $T$ such that for every $t>T$ the sign of a node 
	$u \in V$ depends only on $\bm{\chi}(u)$, i.e., 
	$\sgn(\bm{x}^{(2t)}(u) - \bm{x}^{(2t+1)}(u)) = \sgn(\alpha_{n}\bm{\chi}(u))$. 
	Since $\abs{\bm{\chi}(u)} = 1$, the last equation holds whenever
	\begin{equation}\label{eq:time_bip}
		2\abs{\alpha_{n}} > \abs{\bm{e}^{(2t)}(u) - \bm{e}^{(2t + 1)}(u)}.
	\end{equation}
	We upper bound $\abs{\bm{e}^{(2t)}(u) - \bm{e}^{(2t + 1)}(u)}$ by using 
	\cref{lemma:norm_of_e}, getting that
	$\abs{\bm{e}^{(2t)}(u) - \bm{e}^{(2t+1)}(u)} \leq 2\lambda_2^{2t}\sqrt{\Delta n}$.
	Therefore, with some algebra we get that \cref{eq:time_bip} holds 
	in every round $t > T$, where $T$ is defined as
	\[
		T:=\frac{\log\left(\sqrt{\Delta n}/\abs{\alpha_{n}}\right)}{2\log(1/\lambda_2)}.
	\]
	To conclude the proof, we provide a lower bound on $\abs{\alpha_{n}}$
	showing that it is not too small, w.h.p.
	Recall that $\alpha_i = \frac{\langle \Dsq \bm{x}, \Dsq \bm{v}_i\rangle}{\norm{\Dsq \bm{v}_i}^2}$ 
	and thus
	\begin{equation}\label{eq:alpha_n}
	\alpha_n 
	= \frac{\langle \Dsq \bm{x}, \Dsq \bm{\chi}\rangle}{\norm{\Dsq \bm{\chi}}^2} 
	= \frac{1}{\vol(V)} \sum_{v \in V}\dv(v)\bm{x}(v)\bm{\chi}(v).
	\end{equation}
	The lower bound then follows, with high probability. Indeed,
	\[
	\Pr\left(\abs{\alpha_{n}} \leq \frac{1}{\Delta n}\right)
	\leq \Pr\left(\abs{\alpha_{n}} \leq \frac{1}{\vol(V)}\right) 
	\stackrel{(a)}{=} \Pr\left(\Abs{\sum_{v \in V}\dv(v)\bm{x}(v)\bm{\chi}(v)}\leq 1\right) 
	\stackrel{(b)}{=} \bigO\left(\frac{1}{\sqrt{n}}\right),
	\] 
	where in $(a)$ we used \cref{eq:alpha_n} and in $(b)$ we applied \cref{thm:littlewood-offord}.
	The thesis then follows from the above bound on $\abs{\alpha_{n}}$
	and from the hypothesis on $\Delta = \bigO(n^K)$, 
	for any arbitrary positive constant $K$.
\end{proof}

\section{Discussion and Outlook}\label{sec:concl}

The focus of this work is on heuristics that implicitely perform 
spectral graph clustering, without explicitely computing the main 
eigenvectors of a matrix describing connectivity properties of the 
underlying network (typically, its Laplacian or a related matrix). In this perspective, 
we extended the work of Becchetti et al. \cite{becchetti2017find} in 
several ways. In particular, for $k$ communities, 
\cite{becchetti2017find} considered an extremely regular case, in which the 
second eigenvalue of the (normalized) Laplacian has algebraic and geometric multiplicities $k-1$ and the
corresponding eigenspace is spanned by a basis of indicator vectors. We 
considered a more general case in which the first $k$ eigenvalues are 
in general different, but the span of the corresponding eigenvectors 
again admits a base of indicator vectors. We also made a connection 
between this stepwise property and lumpability properties of the 
underlying random walk, which results in a class of volume-regular 
graphs, that may not have constant degree, nor exhibit balanced communities.
We further showed that our approach naturally lends itself to 
addressing related, yet different problems, such as identifying 
bipartiteness. Finally, in the paragraphs that follow we discuss 
extensions to slightly more general classes than the ones considered in 
this work. 

\paragraph{Other graph classes.}
Consider $k$-volume regular graphs whose $k$ stepwise 
eigenvectors are associated to the $k$ largest eigenvalues, in absolute value.
These graphs include many $k$-partite graphs (e.g., regular ones),
graphs that are ``close'' to being $k$-partite 
(i.e., ones that would become $k$-partite upon removal of a few edges).
Differently from the clustered case (\cref{thm:main})
some of the $k$ eigenvalues can in general be negative.

Consider the following variant of the labeling scheme of the 
\averaging{} dynamics, in which nodes apply their labeling rule only on 
even rounds, comparing their value with the one they held at the end of 
the last even round, i.e., each node $v \in V$ sets $\texttt{label}^{(2t)}(v) = 1$ 
if $\bx^{(2t)}(v) \geqslant \bx^{(2t-2)}(v)$ 
and $\texttt{label}^{(2t)}(v) = 0$ otherwise.
Since the above protocol amounts to only taking even powers of 
eigenvalues, the analysis of this modified protocol proceeds along the same lines as the
clustered case, while the results of \cref{thm:main} seamlessly extend to this class of graphs.

\paragraph{Outlook.}
Though far from conclusive, we believe our results point to 
potentially interesting directions for future research. In general, our 
analysis sheds further light on the connections between temporal 
evolution of the power method and spectral-related clustering 
properties of the underlying network. At the same time, we showed that variants of the \avg{} dynamics 
(and/or its labeling rule) might be useful in addressing different 
problems and/or other graph classes, as the examples given in Section \ref{sec:bipartite} 
suggest. On the other hand, identifying $k$ hidden partitions using the algorithm presented in 
\cite{becchetti2017find} requires relatively strong assumptions on the 
$k$ main eigenvalues and knowledge of an upper bound to the graph 
size,\footnote{As anecdotal experimental evidence suggests, the 
presence of a time window to perform labeling is not an artifact of our analysis.} 
while the analysis becomes considerably more intricate 
than the perfectly regular and completely balanced case 
addressed in \cite{becchetti2017find}. Some aspects of our analysis
(e.g., the aforementioned presence of a size-dependent time window in which the labeling 
rule has to be applied) suggest that more sophisticated variants of the 
\avg{} dynamics might be needed to express the full power of a spectral 
method that explicitely computes the $k$ main 
eigenvectors of a graph-related matrix. While we believe this goal 
can be achieved, designing and analyzing such an algorithm might prove a
challenging task.

\appendix
\section*{Appendix}
\section{Useful inequalities}\label{sec:apx:background}
\begin{theorem}[Extension of Chernoff Bounds~\cite{dubhashi2009concentration}]\label{thm_chernoff_ext}
	Let $X = \sum_{i=1}^n X_i$ where $X_i$ are independent distributed random variables 
	taking values in $\{0,1\}$ and let $\mu = \Ex{X}$. 
	Suppose that $\mu_L \leqslant \mu \leqslant \mu_H$. 
	Then, for $0 <\delta < 1$,
	\[
		\Prob{}{X > (1 + \delta)\mu_H} \leqslant \exp\left(- \frac{\delta^2}{3}\mu_H\right),
	\]
	\[
		\Prob{}{X < (1 - \delta)\mu_L} \leqslant \exp\left(- \frac{\delta^2}{2}\mu_L\right).
	\]
\end{theorem}

\begin{theorem}[Cauchy-Schwarz's inequality]\label{thm:cauchy-schwarz}
	For all vectors $\bm{u}, \bm{v}$ of an inner product space it holds that
	\(
	|\langle\bm{u}, \bm{v}\rangle|^2
	\leq \langle{\bm{u}, \bm{u}}\rangle \cdot \langle{\bm{v}, \bm{v}}\rangle,
	\)
	where $\langle{\cdot, \cdot}\rangle$ is the inner product.
\end{theorem}

\begin{theorem}[Cheeger's inequality~\cite{chung96}]\label{thm:cheeger}
	Let $P$ be the transition matrix of a connected edge-weighted graph $G=(V,E,w)$
	and let $\lambda_2$ be its second largest eigenvalue.
	Let $\abs{E(S, V \setminus S)} = \sum_{u\in S,\, v \in V\setminus S} w(u,v)$
	and
	\(
	h_G = \min_{S : \vol(S) \leq \frac{\vol(V)}{2}}
	\frac{\abs{E(S, V \setminus S)}}{\vol(S)}.
	\)
	Then
	\[
		\frac{1-\lambda_2}{2} \leq h_G \leq \sqrt{2 (1 - \lambda_2)}.
	\]
\end{theorem}

\begin{theorem}[Berry-Esseen's theorem \cite{berry1941accuracy}]\label{thm:berry-esseen_non_iid}
	Let $X_1,\ldots,X_n$ be independent random variables
	with mean $\mu_i=0$, variance $\sigma_i^2 > 0$, and third absolute moment
	$\rho_i < \infty$, for every $i=1,\ldots,n$.
	Let $S_n = \sum_{i=1}^{n}X_i$ and let $\sigma = \sqrt{\sum_{i=1}^n \sigma_i^2}$
	be the standard deviation of $S_n$;
	let $F_n$ be the cumulative distribution function of $\frac{S_n}{\sigma}$;
	let $\Phi$ the cumulative distribution function of the standard normal distribution.
	Then, there exists a positive constant $C$ such that, for all $x$ and for all $n$,
	\[
		\vert F_n(x) - \Phi(x) \vert \leq \frac{C\psi}{\sigma},
	\]
	where $\psi := \max_{i\in\{1,\ldots,n\}} \frac{\rho_i}{\sigma_i^2}$.
\end{theorem}

\begin{theorem}[Littlewood-Offord's small ball~\cite{erdos1945lemma}]\label{thm:littlewood-offord}
	Let $a_1, \dots, a_n \in \mathbb{R}$ be real numbers with $|a_i| \geq 1$ 
	for every $i=1, \dots, n$ and let $r \in \mathbb{R}$ be any real number.  Let
	$\{X_i \,:\, i = 1, \dots, n \}$ be a family of independent Rademacher random
	variables (taking values $\pm 1$ with probability $1/2$) and let $X$ be their
	sum weighted with the $a_i$s, i.e., $X = \sum_{i=1}^{n} a_i X_i$, then
	\[
		\Pr(\Abs{X - r} < 1) = \bigO\left(\frac{1}{\sqrt{n}}\right).
	\]
\end{theorem}

\bibliographystyle{alpha}
\bibliography{community}

\end{document}